\documentclass[oneside,11pt,a4paper,article]{memoir}
\settypeblocksize{23cm}{15cm}{*}
\setbinding{0cm}
\setulmargins{*}{*}{1}
\setlrmargins{*}{*}{1}
\checkandfixthelayout

\usepackage{microtype}
\title{An All-But-One Entropic Uncertainty Relation, and Application to Password-Based Identification}
\author{Niek J. Bouman$^{1}$, Serge Fehr$^{1}$,\\Carlos Gonz\'alez-Guill\'en$^{2,3}$ and Christian Schaffner$^{4,1}$\\
\small \emph{$^1$Centrum Wiskunde \& Informatica (CWI), Amsterdam, The Netherlands}\\
\small \emph{$^2$Universidad Polit\'ecnica de Madrid, Spain \hspace{.5em} %}\\
$^3$IMI, Universidad Complutense de Madrid, Spain}\\
\small\emph{$^4$University of Amsterdam (UvA), The Netherlands}}

\date{}

\usepackage{amssymb}
\usepackage{amsmath}
\usepackage{amsfonts}
\usepackage{mathrsfs}
\usepackage{amsthm}
\usepackage{xspace}
\usepackage{comment}
\usepackage{booktabs}
\usepackage{graphicx}
\usepackage{ifpdf}
\ifpdf
\usepackage[usenames,dvipsnames]{color}
\else
\usepackage[usenames,dvips]{color}
\fi
\usepackage{mdwlist}

\usepackage{verbatim}
\usepackage{ifthen}

\usepackage{hyperref}
\usepackage{float}
\floatstyle{ruled}
\newfloat{protocol}{thp}{lop}
\floatname{protocol}{Protocol}
\definecolor{RemarkRed}{rgb}{0.8,0.0,0.0}

\newcommand{\co}{0}
\newcommand{\cl}{1}
\ifpdf

\else

\fi

\newcommand{\bias}{\ensuremath{\mathrm{bias}}}

\renewcommand{\Pr}{\ensuremath{\mathrm{Pr}}}
\newcommand{\nbit}{\{ 0,1 \} ^n}

\renewcommand{\P}{P}
\newcommand{\ol}[1]{\overline{#1}}

\newcommand{\out}[2]{|#1\rangle \!\langle #2|\xspace}
\newcommand{\outs}[1]{|#1\rangle \!\langle #1|\xspace}
\newcommand{\inprod}[2]{\langle #1|#2 \rangle\xspace}

\newcommand{\tvec}[1]{\langle #1 |\xspace}
\newcommand{\kron}{\otimes\xspace}
\newcommand{\xor}{\oplus\xspace}
\newcommand{\schur}{\odot\xspace}
\newcommand{\ket}[1]{\ensuremath{| #1 \rangle}\xspace}         % What I'm used to...
\newcommand{\bra}[1]{\tvec{#1}}
\newcommand{\braket}[2]{\inprod{#1}{#2}}
\newcommand{\ketbra}[2]{\out{#1}{#2}}
\newcommand{\proj}[1]{\outs{#1}}
\newcommand{\tr}{\mathrm{tr}}

\newcommand{\trace}{\mathrm{tr}}
\newcommand{\I}{\mathbb{I}}
\newcommand{\dist}{\delta}

\newcommand{\set}[1]{\{#1\}}
\newcommand{\Set}[2]{\{#1:#2\}}
\newcommand{\setn}[1][n]{[#1]}
\newcommand{\cnum}{\mathbb{C}}
\newcommand{\bat}{\{0,1\}}
\newcommand{\mcal}[1]{\mathcal{#1}}

\newcommand{\sdm}{\mcal{D}}
\newcommand{\R}{\mathbb{R}}
\newcommand{\N}{\mathbb{N}}

\newcommand{\encoding}{\mathfrak{c}\xspace}

\newcommand{\Tset}[1]{\ensuremath{\mathcal{#1}}\xspace}

\newcommand{\etal}{{\emph{et al.}}\xspace}

\newcommand{\assign}{\ensuremath{:=}}

\newcommand{\B}{{\mathcal{B}}}

\newcommand{\E}{\mathcal{E}}
\newcommand{\C}{\mathcal{C}}
\renewcommand{\L}{\mathcal{L}}

\newcommand{\gs}{p_\mathrm{guess}}

\newcommand{\QID}{{\bfseries\texttt{Q-ID}}\xspace}

\newcommand{\SD}{\mathrm{SD}}

\newcommand{\stdist}{\mathrm{SD}}

\newcommand{\distuni}{d_\mathrm{unif}}
\newcommand{\q}[1]{\hat{#1}}   % quantized basis

\newcommand{\refsec}[1]{Section~\ref{sec:#1}\xspace}
\newcommand{\refthm}[1]{Theorem~\ref{thm:#1}\xspace}
\newcommand{\reflem}[1]{Lemma~\ref{lem:#1}\xspace}
\newcommand{\refcor}[1]{Corollary~\ref{cor:#1}\xspace}

\newcommand{\refapp}[1]{Appendix~\ref{app:#1}\xspace}
\newcommand{\refdef}[1]{Definition~\ref{def:#1}\xspace}
\newcommand{\refprop}[1]{Proposition~\ref{prop:#1}\xspace}

\newcommand{\linspan}{\ensuremath{\mathrm{span}}}
\newcommand{\lswon}[1][F]{\ensuremath{\linspan({#1})\setminus \set{0^n}}\xspace} % linspan w/o null
\newcommand{\hmin}[1][]{\ensuremath{H^{#1}_\mathrm{min}\hspace{-1pt}}}
\newcommand{\hmax}{\ensuremath{H_\mathrm{max}\hspace{-1pt}}}

\newcommand{\id}{\ensuremath{\mathbb{I}}\xspace}

\renewcommand{\H}{\mathcal{H}}

\newcommand{\refeq}[1]{(\ref{eq:#1})\xspace}

\theoremstyle{plain} \newtheorem{thm}{Theorem} %[chapter]
\theoremstyle{plain} \newtheorem{lemma}[thm]{Lemma} %[chapter]
\theoremstyle{plain} \newtheorem{corollary}[thm]{Corollary} %[chapter]
\theoremstyle{plain} \newtheorem{prop}[thm]{Proposition} %[chapter]

\theoremstyle{definition} \newtheorem{definition}[thm]{Definition}

\newtheoremstyle{example}{\topsep}{\topsep}%
     {}%         Body font
     {}%         Indent amount (empty = no indent, \parindent = para indent)
     {\bfseries}% Thm head font
     {.}%        Punctuation after thm head
     {\newline}%     Space after thm head (\newline = linebreak)
     {\thmname{#1}\thmnumber{ #2}\thmnote{ {\normalfont(#3)}}}%         Thm head spec

   \theoremstyle{example}
\definecolor{darkgreen}{rgb}{0,0.6,0}

\newcommand{\user}{\ensuremath{\mathsf{U}}\xspace}
\newcommand{\server}{\ensuremath{\mathsf{S}}\xspace}
\newcommand{\dishuser}{\ensuremath{\mathsf{U}^*}\xspace}
\newcommand{\dishserver}{\ensuremath{\mathsf{S}^*}\xspace}

\begin{document}
%%%%%%%%%%%%%%%%%%%%%%%%%%%%%%%%%%%%%%%%%%%%%%%%%%%%%%%%%%%%%%
%%%%%%%%%%%%%%%%%%%%%%%%%%%%%%%%%%%%%%%%%%%%%%%%%%%%%%%%%%%%%%

\bibliographystyle{alpha}

%%%%%%%%%%%%%%%%%%%%%%%%%%%%%%%%%%%%%%%%%%%%%%%%%%%%%%%%%%%%%%
\maketitle

\begin{abstract} 
Entropic uncertainty relations are quantitative characterizations of Heisenberg's uncertainty principle, which make use of an entropy measure to quantify uncertainty. 
In quantum cryptography, they are often used as convenient tools in security proofs. 

We propose a new entropic uncertainty relation. It is the first such uncertainty relation that lower bounds 
%the uncertainty of {\em all but one} measurement outcomes 
the uncertainty in the measurement outcome for \emph{all but one} choice for the measurement from %zwith respect to 
an {\em arbitrarily large} (but specifically chosen) set of possible measurements, and, at the same time, uses the {\em min-entropy} as entropy measure, rather than the Shannon entropy. This makes it especially suited for quantum cryptography.

As application, we propose a new {\em quantum identification scheme} in the bounded-quantum-storage model. Because the scheme requires a perfectly single-qubit source to operate securely, it is currently mainly of theoretical interest. Our new uncertainty relation forms the core of the new scheme's security proof in the bounded-quantum-storage model. In contrast to the original quantum identification scheme proposed by Damg{\aa}rd \etal, our new scheme also offers some security in case the bounded-quantum-storage assumption fails to hold. Specifically, our scheme remains secure against an adversary that has unbounded storage capabilities but is restricted to non-adaptive single-qubit operations. The scheme by Damg{\aa}rd \etal, on the other hand, completely breaks down under such an attack. 
\vspace{1.5em}\\
\footnotesize NJB is supported by an NWO Open Competition grant. CGG
is supported by Spanish Grants I-MATH, MTM2008-01366, QUITEMAD and
QUEVADIS. CS is supported by an NWO VENI grant.
\end{abstract}

\newpage
\tableofcontents*
\newpage

\chapter{Introduction }
\section{A New Uncertainty Relation}
\label{sec:newuncert}
In this work, we propose and prove a new general entropic uncertainty relation. Uncertainty relations are quantitative characterizations of the uncertainty principle of quantum mechanics, which expresses that for certain pairs of measurements, there exists no state for which the measurement outcome is determined for {\em both} measurements: at least one of the outcomes must be somewhat uncertain. {\em Entropic} uncertainty relations express this uncertainty in at least one of the measurement outcomes by means of an entropy measure, usually the Shannon entropy. 
Our new entropic uncertainty relation distinguishes itself from previously known uncertainty relations by the following collection of features:

\begin{enumerate}
\item\label{it:minentropy} It uses the \emph{min-entropy} as entropy measure, rather than the Shannon entropy. Such an uncertainty relation is sometimes also called a \emph{high-order} entropic uncertainty relation.%
\footnote{This is because the min-entropy coincides with the R{\'e}nyi entropy $H_\alpha$ of high(est) order $\alpha = \infty$. In comparison, the Shannon entropy coincides with the R{\'e}nyi entropy of (relatively) low order $\alpha = 1$.} Since privacy amplification needs a lower bound on the min-entropy, high-order entropic uncertainty relations are useful tools in quantum cryptography. 
%  }
%Since the min-entropy is a more conservative measure (i.e., never larger than the Shannon entropy), it guarantees a \emph{stronger} type of uncertainty. 

\item\label{it:allbutone} 
It lower bounds the uncertainty in the measurement outcome for \emph{all but one} measurement, chosen from an \emph{arbitrary} (and arbitrarily large) family of possible measurements. This is clearly \emph{stronger} than typical entropic uncertainty relations that lower bound the uncertainty on \emph{average} (over the choice of the measurement).

\item\label{it:qubitwise} The measurements can be chosen to be qubit-wise measurements, in the computational or Hadamard basis, and thus the uncertainty relation is applicable to practical schemes (which can be implemented using current technology). 
\end{enumerate}

To the best of our knowledge, no previous entropic uncertainty relation satisfies (\ref{it:minentropy}) and (\ref{it:allbutone}) simultaneously, let alone in combination with~(\ref{it:qubitwise}). Indeed, as pointed out in a recent overview article by Wehner and Winter~\cite{WW10}, little is known about entropic uncertainty relations for more than two measurement outcomes, and even less when additionally considering min-entropy. 

%\subsection{Explanation by means of a Simpler Entropic Uncertainty Relation}% Explained}
To explain our new uncertainty relation, we find it helpful to first discuss a simpler variant, which does not satisfy~(\ref{it:minentropy}), and which follows trivially from known results. 
Fix an arbitrary family $\set{\B_1,\ldots,\B_m}$ of bases for a given quantum system (i.e., Hilbert space). The \emph{maximum overlap} of such a family is defined as 
\[
c := \max\Set{|\braket{\phi}{\psi}|}{\ket{\phi} \in \B_j, \ket{\psi} \in \B_k, 1 \!\leq\! j \!<\! k \!\leq\! m},
\]
and let $d := -\log(c^2)$. 
Let $\rho$ be an arbitrary quantum state of that system, and let $X$ denote the measurement outcome when $\rho$ is measured in one of the bases. We model the choice of the basis by a random variable $J$, so that $H(X|J\!=\!j)$ denotes the Shannon entropy of the measurement outcome when $\rho$ is measured in basis $\B_j$. 
It follows immediately from Maassen and Uffink's uncertainty relation~\cite{MU88} that
\[
H(X|J =j)+H(X|J=k) \geq -\log(c^2) = d \quad \forall j \neq k.
\]
As a direct consequence, there exists a choice $j'$ for the measurement so that $H(X|J\!=\!j) \geq \frac{d}{2}$ for all $j \in \set{1,\ldots,m}$ with $j \neq j'$. In other words, for any state $\rho$ there exists $j'$ so that unless the choice for the measurement coincides with $j'$, which happens with probability at most $\max_j P_J(j)$, there is at least $d/2$ bits of entropy in the outcome~$X$.

Our new high-order entropic uncertainty relation shows that this very statement essentially still holds when we replace Shannon by min-entropy, except that $j'$ becomes randomized: for any $\rho$, there exists a \emph{random variable} $J'$, independent of $J$, such that%
\footnote{The rigorous version of the approximate inequality $\gtrsim$ is stated in \refthm{UR}. } 
$$
\hmin(X|J\!=\!j, J' \!=\! j') \gtrsim \frac{d}{2} \quad \forall \; j \neq j' \in \set{1,\ldots,m} 
$$
no matter what the distribution of $J$ is. 
Thus, unless the measurement $J$ coincides with $J'$, there is roughly $d/2$ bits of min-entropy in the outcome~$X$. Furthermore, since $J'$ is \emph{independent} of $J$, the probability that $J$ coincides with $J'$ is at most $\max_j P_J(j)$, as is the case for a fixed $J'$. 

Note that we have no control over (the distribution of) $J'$. We can merely guarantee that it exists and is independent of $J$. It may be insightful to interpret $J'$ as a \emph{virtual guess} for $J$, guessed by the party that prepares~$\rho$, and whose goal is to have little uncertainty in the measurement outcome $X$. 
%Then, our uncertainty relation guarantees that there is lots of uncertainty in $X$ (if $d$ is large), unless the guess was correct. 
The reader may think of the following specific way of preparing $\rho$: sample $j'$ according to some arbitrary distribution $J'$, and then prepare the state as the, say, first basis vector of $\B_{j'}$. If the resulting mixture $\rho$ is then measured in some basis  $\B_j$, sampled according to an arbitrary (independent) distribution~$J$, then unless $j = j'$ (i.e., our guess for $j$ was correct), there is obviously lower bounded uncertainty in the measurement outcome $X$ (assuming a non-trivial maximum overlap). 
Our uncertainty relation can be understood as saying that for \emph{any} state~$\rho$, no matter how it is prepared, there exists such a (virtual) guess $J'$, which exhibits this very behavior: if it differs from the actual choice for the measurement then there is lower bounded uncertainty in the measurement outcome $X$. 
As an immediate consequence, we can for instance say that $X$ has min-entropy at least $d/2$, except with a probability that is given by the probability of guessing $J$, e.g., except with probability $1/m$ if the measurement is chosen uniformly at random from the family. 
%This in particular implies that the min-entropy of $X$ is lower bounded except with the probability of being able to guess $J$. 
This is clearly the best we can hope for. 

We stress that because the min-entropy is more conservative than the Shannon entropy, our high-order entropic uncertainty relation does not follow from its simpler Shannon-entropy version. Neither can it be deduced in an analogous way; the main reason being that for fixed pairs $j \neq k$, there is no strong lower bound on $\hmin(X|J\!=\!j)+\hmin(X|J\!=\!k)$, in contrast to the case of Shannon entropy. More precisely and more generally, the \emph{average} uncertainty $\frac{1}{|J|}\sum_j\hmin(X|J\!=\!j)$ does not allow a lower bound higher than $\log|J|$. To see this, consider the following example for $|J|=2$ (the example can easily be extended to arbitrary $|J|$). Suppose that $\rho$ is the uniform mixture of two pure states, one giving no uncertainty when measured in basis $j$, and the other giving no uncertainty when measured in basis $k$. Then, $\tfrac12 \hmin(X|J\!=\!j) + \tfrac12 \hmin(X|J\!=\!k) = 1$. 
Because of a similar reason, we cannot hope to get a good bound for all but a {\em fixed} choice of $j'$; the probabilistic nature of $J'$ is necessary (in general).  
Hence, compared to bounding the average uncertainty, % $\frac{1}{|J|}\sum_j\hmin(X|J\!=\!j)$, 
the all-but-one form of our uncertainty relation not only makes our uncertainty relation stronger in that uncertainty for all-but-one implies uncertainty on average (yet not vice versa), but it also allows for {\em more} uncertainty.

By using asymptotically good error-correcting codes, one can construct families of bases that have a large value of $d$, and thus for which our uncertainty relation guarantees a large amount of min-entropy (we discuss this in more detail in \refsec{goodfam}). These families consist of qubit-wise measurements in the computational or the Hadamard basis, hence these measurements can be performed with current technology. %in practice.

The proof of our new uncertainty relation comprises a rather involved probability reasoning to prove the existence of the random variable $J'$ and builds on earlier work presented in \cite{Schaffner07}. 

\section{Quantum Identification with ``Hybrid'' Security}
As an application of our entropic uncertainty relation, we propose a new \emph{quantum identification protocol}. Informally, the goal of (password-based) identification is to prove knowledge of a possibly low-entropy password $w$, without giving away any information on $w$ (beyond what is unavoidable). In~\cite{DFSS07}, Damg{\aa}rd \etal\ showed the existence of such an identification protocol%
%\footnote{Actually, \cite{DFSS07} proposed \emph{two} such schemes: \protfont{QID} and \protfont{QID$^+$}. \QID offers security against \emph{impersonation attacks}, and \xQID additionally offers security against \emph{man-in-the-middle attacks} but is not truly password-based. In this work, we focus on impersonation attacks only (with truly password-based security). } 
in the \emph{bounded-quantum-storage model} (BQSM). This means that the proposed protocol involves the communication of qubits, and security is proven against any dishonest participant that can store only a limited number of these qubits (whereas legitimate participants need no quantum storage at all to honestly execute the protocol). 
%\end{comment}

Our uncertainty relation gives us the right tool to prove security of the new quantum identification protocol in the BQSM. The distinguishing feature of our new protocol is that it also offers some security in case the assumption underlying the BQSM fails to hold. Indeed, we additionally prove security of our new protocol against a dishonest server that has unbounded quantum-storage capabilities and can reliably store all the qubits communicated during an execution of the protocol, but is restricted to non-adaptive single-qubit operations and measurements.\footnote{It is known that \emph{some} restriction is necessary (see \cite{DFSS07}).} 
% (i.e., cannot operate on several qubits coherently).%
This is in sharp contrast to protocol \textsf{QID} by Damg{\aa}rd \etal, which completely breaks down against a dishonest server that can store all the communicated qubits in a quantum memory and postpone the measurements until the user announces the correct measurement bases. On the downside, our protocol only offers security in case of a perfectly single-qubit (e.g.\ single-photon) source, because multi-qubit emissions reveal information about $w$. Hence, given the immature state of single-qubit-source technology at the time of this writing, our protocol is currently mainly of theoretical interest. 

We want to stress that proving security of our protocol in this \emph{single-qubit-operations model} (SQOM) is non-trivial. Indeed, as we will see, standard tools like privacy amplification are not applicable. Our proof relies on a certain minimum-distance property of random binary matrices 
%involves certain properties of random linear codes 
and makes use of Diaconis and Shahshahani's XOR inequality (\refthm{diaconis}, see also \cite{Diaconis88}). 

\section{Related Work}

The study of \emph{entropic} uncertainty relations, whose origin dates back to 1957 with the work of Hirschman~\cite{Hirschman57}, has received a lot of attention over the last decade due to their various applications in quantum information theory. We refer the reader to~\cite{WW10} for a recent overview on entropic uncertainty relations. Most of the known entropic uncertainty relations are %expressed by means of Shannon entropy and/or are 
of the form 
$$
\frac{1}{|J|}\sum_j H_\alpha(X|J\!=\!j) \geq h \, ,
$$
where $H_\alpha$ is the R\'enyi entropy.\footnote{The R\'enyi entropy \cite{renyi1961} is defined as $H_\alpha(X) := \frac{1}{1-\alpha}\log\sum_x P_X(x)^\alpha$.  Nevertheless, for most known uncertainty relations $\alpha=1$, i.e.\ the Shannon entropy.} I.e., most uncertainty relations only give a lower bound on the entropy
%\footnote{For different $\alpha$'s, $H_\alpha$ corresponds to different entropy measures, e.g., $H_1 = H$ (the Shannon entropy) and $H_\infty = \hmin$. } 
of the measurement outcome $X$ \emph{on average} over the (random) choice of the measurement. 
As argued in \refsec{newuncert}, %For such an uncertainty relation, 
the bound $h$ on the \emph{min}-entropy can be at most 
%not be large, it can be at most 
$\log|J|$, no matter the range of~$X$. Furthermore, an uncertainty relation of this form only guarantees that there is uncertainty in $X$ for \emph{some} measurement(s), but does not specify precisely for how many, and certainly it does not guarantee uncertainty for \emph{all but one} measurements. 
The same holds for the high-order entropic uncertainty relation from~\cite{dfrss07}, which considers an exponential number of measurement settings and guarantees that except with negligible probability over the (random) choice of the measurement, there is lower-bounded min-entropy in the outcome. On the other hand, the high-order entropic uncertainty relation from~\cite{DFSS05} only considers \emph{two} measurement settings and guarantees lower-bounded min-entropy with probability (close to) $\frac{1}{2}$. 

The uncertainty relation we know of that comes closest to ours is Lemma~2.13 in~\cite{FHS11}. Using our notation, it shows that $X$ is $\epsilon$-close to having roughly $d/2$ bits of min-entropy (i.e., the same bound we get), 
% gives a lower bound on the $\epsilon$-smooth min-entropy $\hmin[\,\epsilon](X|C\!=\!c)$ of roughly $d/2$ too, 
but only for all but an $\epsilon$-fraction of all the $m$ possible choices for the measurement $j$, where $\epsilon$ is about~\smash{$\sqrt{2/m}$}. 

With respect to our application, backing up the security of the identification protocol by Damg{\aa}rd \etal~\cite{DFSS07} against an adversary that can overcome the quantum-memory bound assumed by the BQSM was also the goal of~\cite{DFLSS09}. However, the solution proposed there relies on an unproven computational-hardness assumption, and as such, strictly speaking, can be broken by an adversary in the SQOM, i.e., by storing qubits and measuring them later qubit-wise and performing (possibly infeasible) classical computations. On the other hand, by \emph{assuming} a lower bound on the hardness of the underlying computational problem against quantum machines, the security of the protocol in~\cite{DFLSS09} holds against an adversary with much more quantum computing power than our protocol in the SQOM, which restricts the adversary to single-qubit operations. 

We hope that with future research on this topic, new quantum identification (or other cryptographic) protocols will be developed with security in the same spirit as our protocol, but with a more relaxed restriction on the adversary's quantum computation capabilities, for instance that he can only perform a limited number of quantum computation steps, and in every step he can only act on a limited number of qubits coherently.

\chapter{Preliminaries}
\section{Basic Notation}
%A matrix is written as a bold capital letter, for example \mx{A}. \serge{Not necessarily in Section~3.} 
%For a matrix $A$, $A^*$ denotes the Hermitian transpose of $A$. 
Sets as well as families are written using a calligraphic font, e.g.\ $\Tset{A}, \Tset{X}$, and we write $|\Tset{A}|$ etc.\ for the cardinality. We use $\setn$ as a shorthand for $\set{1,\ldots,n}$. 
%For two sets \Tset{A} and \Tset{B}, we denote their union as $\Tset{A} \cup \Tset{B}$, their intersection as $\Tset{A} \cap \Tset{B}$, their difference as $\Tset{A} \setminus \Tset{B}$ and their symmetric difference as $\Tset{A} \symdif \Tset{B} := (\Tset{A} \setminus \Tset{B} ) \cup (\Tset{B}\setminus \Tset{A})$.
%\serge{As far as I can see, we're using this symmetric difference notion {\em once} (in a technical proof); hence, no need to introduce it and give it a name here. The other set operations are so common that I don't think we need to explicitly introduce them.}\niek{indeed.}

%For an $n$-bit vector vector $v = (v_1,\ldots,v_n)$ in $\set{0,1}^n$, we write $|v|$ for its Hamming weight, and, for any subset $\Tset{I} \subseteq \setn$, we write $v_{\Tset{I}}$ for the restricted vector $(v_i)_{i\in \Tset{I}} \in \set{0,1}^{|\Tset{I}|}$. If $w$ is an additional vector in $\set{0,1}^n$, the \emph{Schur product} is defined as the element-wise product $v \schur w := (v_1 w_1, v_2 w_2, \ldots, v_n w_n) \in \set{0,1}^n$, and the inner product between $v$ and $w$ is given by $v\cdot w := v_1 w_1 \oplus \cdots \oplus v_n w_n \in \set{0,1}$, where the addition is modulo~$2$. 

%the image of $F$ when seen as function $\set{0,1}^\ell \to \set{0,1}^n$, $v \mapsto vF$.
%\niek{This is the new lemma:}
For an $n$-bit vector vector $v = (v_1,\ldots,v_n)$ in $\set{0,1}^n$, we write $|v|$ for its Hamming weight, and, for any subset $\Tset{I} \subseteq \setn$, we write $v_{\Tset{I}}$ for the restricted vector $(v_i)_{i\in \Tset{I}} \in \set{0,1}^{|\Tset{I}|}$.
For two vectors $v,w \in \bat^n$,
%If $w$ is an additional vector in $\set{0,1}^n$, 
the \emph{Schur product} is defined as the element-wise product $v \schur w := (v_1 w_1, v_2 w_2, \ldots, v_n w_n) \in \set{0,1}^n$, and the \emph{inner product} between $v$ and $w$ is given by $v\cdot w := v_1 w_1 \oplus \cdots \oplus v_n w_n \in \set{0,1}$, where the addition is modulo~$2$. 
We write $\linspan(F)$ for the \emph{row span} of a matrix $F$; the set of vectors obtained by making all possible linear combinations (modulo~$2$) of the rows of $F$, i.e. 
the set $\Set{sF}{\forall s\in \bat^\ell}$, where $s$ should be interpreted as a row vector and $sF$ denotes a vector-matrix product. 

\section{Probability Theory} %Random Variables, Probability Distributions, Markov Chains and Bias}
A finite probability space is a non-empty finite set $\Omega$ together with a function $\Pr: \Omega \rightarrow \mathbb{R}$ such that $\Pr(\omega)\geq 0 \quad \forall \omega \in \Omega$ and $\sum_{\omega\in \Omega}\Pr(\omega)=1$. An \emph{event} is a subset of $\Omega$. 
A {\em random variable} is a function $X: \Omega \rightarrow \mathcal{X}$ from a finite probability space %\footnote{with the power set as its measure space.} 
$(\Omega,\Pr)$ to a finite set $\mathcal{X}$. 
We denote random variables as capital letters, for example $X$, $Y$, $Z$. %Note that we have already used bold capital letters for matrices. If we want to denote a \emph{random} matrix, we will use italic style and write $\boldsymbol{A}$. 
The {\em distribution} of $X$, which we denote as $P_X$, is given by $P_X(x) = \Pr[X\!=\!x] = \Pr[\Set{\omega \in \Omega}{X(\omega)\!=\!x}]$. The joint distribution of two (or more) random variables $X$ and $Y$ is denoted by $P_{XY}$, i.e., $P_{XY}(x,y) = \Pr[X\!=\!x \wedge Y\!=\!y]$. Specifically, we write $U_{\Tset{X}}$ for the uniform probability distribution over \Tset{X}. Usually, we leave the probability space $(\Omega,\Pr)$ implicit, and understand random variables to be defined by their joint distribution, or by some ``experiment'' that uniquely determines their joint distribution. %A probability is said to be negligible in $n$, denoted as $\negl(n)$ if for any polynomial $p$, it is smaller than $1/p(n)$ for all sufficiently large $n$. \serge{Are we using this?}\niek{only in the Correctness definition}

Random variables $X$ and $Y$ are {\em independent} if $P_{XY} = P_X P_Y$ (which should be understood as $P_{XY}(x,y) = P_X(x) P_Y(y) \;\forall\, x\in {\mathcal{X}},y\in \mathcal{Y}$). 
%\serge{Can we restrict to the case $n = 3$?}
%\niek{Like this?}
The random variables $X$, $Y$ and $Z$ form a (first-order) Markov chain, denoted by $X \leftrightarrow Y \leftrightarrow Z$, if $P_{XZ|Y} = P_{X|Y}P_{Z |Y}$.
%A (first-order) Markov chain is a sequence of random variables $X_1, \ldots, X_n$ such that for all $i\in [n]\setminus\set{1}$, $P_{X_i|X_{i-1},\ldots,X_1} = P_{X_i|X_{i-1}}$, and is denoted by $X_1 \leftrightarrow X_2 \leftrightarrow \cdots  \leftrightarrow X_n$. \serge{Can we restrict to the case $n = 3$?}
The \emph{statistical distance} (also knows as variational distance) between distributions $P_X$ and $P_Y$ is written as $\SD(P_X,P_Y):=\tfrac12\|P_X-P_Y\|_1$. 

%\label{def:bias}
The \emph{bias} of a binary random variable $X$ is defined as
$\mathrm{bias}(X) := \big|P_X(0) - P_X(1) \big|.$ This also naturally defines the bias of $X$ conditioned on an event $\mathcal{E}$ as
$\mathrm{bias}(X|\mathcal{E}) := \big|P_{X|\mathcal{E}}(0) - P_{X|\mathcal{E}}(1) \big|$.  The bias thus ranges between $0$ and $1$ and can be understood as a degree of predictability of a bit: if the bias is small then the bit is close to random, and if the bias is large (i.e. approaches $1$) then the bit has essentially no uncertainty. 
For a sum of two independent binary random variables $X_1$ and $X_2$, the
bias of the sum is the product of the individual biases, i.e. $\bias(X_1 \xor X_2) = \bias(X_1) \bias(X_2)$.
\begin{thm}[Diaconis and Shahshahani's Information-Theoretic XOR Lemma]
\label{thm:diaconis}
Let $X$ be a random variable over $\Tset{X}:=\set{0,1}^n$ with distribution $P_{X}$. Then, the following holds,
\[
\SD(P_{X}, U_\mathcal{X}) \leq \frac12 \Big[\sum_{f \in \set{0,1}^n \setminus \set{0^n}} \bias(f\cdot X)^2\Big]^\frac12.
\]
\end{thm}

\noindent The original version of \refthm{diaconis} appeared in \cite{Diaconis88}, where it is expressed in the language of representation theory. The version above is due to \cite{naor93}.%\footnote{Note that in \cite{naor93}, the bias of the empty set is implicitly understood to be zero, which is actually non-standard. Interestingly, this is actually defined \emph{explicitly} in a preliminary version of the paper, but seems to be accidentally removed in the later version. In our definition of the bias (\refdef{bias}), taking $f$ to be the zero vector yields $\bias(f\cdot X)=1$. Hence, we explicitly exclude the zero vector from the summation in \refthm{diaconis}.}

\begin{thm}[Hoeffding's Inequality]
Let $X_1, X_2, \ldots, X_n$ be independent binary random variables, each distributed according to the Bernoulli distribution with parameter $\mu$, and let $\bar X := n^{-1} \sum_{i \in \setn} X_i$.
%$0 \leq X_i \leq 1$ for all $i \in \setn$. 
Then for $0 < t < 1-\mu$
\[
\Pr[\bar X - \mu \geq t] \leq \exp(-2nt^2).
\] 
\label{thm:hoeffding}
\end{thm}
\noindent For a proof, the reader is referred to \cite{hoeffding1963}. 

\section{Quantum Systems and States}
%%%%%%%%%%%%%%%%%%%%%%%%%%%%%%%%%%%%%%%%%%%%%%%%%%%%%%%%%%%%%%
We assume that the reader is familiar with the basic concepts of quantum information theory; the main purpose of this section is to fix some terminology and notation. 
A quantum system $A$ is associated with a complex Hilbert space, $\H = \mathbb{C}^d$, its {\em state space}. %The {\em state} of $A$ is given, in the case of a {\em pure} state, by a state vector $\ket{\varphi} \in \H$ with norm $\| \ket{\varphi} \| = \sqrt{\braket{\varphi}{\varphi}} = 1$, respectively, in the case of a {\em mixed} state, by a trace-$1$ positive-semi-definite operator/matrix $\rho: \H \rightarrow \H$. In order to simplify language, we are sometimes a bit sloppy in distinguishing between a quantum system, its state, and the state vector or density matrix describing the state. 
By default, we write $\H_A$ for the state space of system $A$, and $\rho_A$ (respectively $\ket{\varphi_A}$ in case of a pure state) for the state of $A$. We write $\sdm(\H)$ for the set of all density matrices on Hilbert space $\H$.

The state space of a {\em bipartite} quantum system $AB$, consisting of two (or more) subsystems, is given by $\H_{AB} = \H_A \otimes \H_B$. If the state of $AB$ is given by $\rho_{AB}$ then the state of subsystem $A$, when treated as a stand-alone system, is given by the {\em partial trace} $\rho_A = \tr_B(\rho_{AB})$, and correspondingly for $B$. 
{\em Measuring} a system $A$ in basis $\set{\ket{i}}_{i \in I}$, where $\set{\ket{i}}_{i \in I}$ is an orthonormal basis of $\H_A$, means applying the measurement described by the projectors $\set{\proj{i}}_{i \in I}$, such that outcome $i \in I$ is observed with probability $p_i = \tr(\proj{i} \rho_A)$ (respectively $p_i = |\braket{i}{\varphi_A}|^2$ in case of a pure state). 
If $A$ is a subsystem of a bipartite system $AB$, then it means applying the measurement described by the projectors $\set{\proj{i} \otimes \I_B}_{i \in I}$, where $\I_B$ is the identity operator on $\H_B$.

A {\em qubit} is a quantum system $A$ with state space~$\H_A = \mathbb{C}^2$. 
%The {\em computational basis} (for a qubit) is denoted as $\set{\ket{\co},\ket{\cl}}$, and the {\em Hadamard basis} as $H\set{\ket{\co},\ket{\cl}} = \set{H\ket{\co},H\ket{\cl}}$, where $H$ is the {\em Hadamard matrix}. 
The {\em computational basis} $\set{\ket{\co},\ket{\cl}}$ (for a qubit) is given by $\ket{\co} = {1 \choose 0}$ and $\ket{\cl} = {0 \choose 1}$, and the {\em Hadamard basis} by $\set{H\ket{\co},H\ket{\cl}}$, where $H$ denotes the 2-dimensional {\em Hadamard matrix} $H = \frac{1}{\sqrt2} \big(\begin{smallmatrix} 1 & \;\; 1 \\ 1 & -1 \end{smallmatrix}\big)$. 
We also call the computational basis the {\em plus} basis and associate it with the `$+$'-symbol, and we call the Hadamard basis the {\em times} basis and associate it with the `$\times$'-symbol. 
%Sometimes we associate the `$+$'-symbol with the computational basis and the `$\times$'-symbol with the Hadamard basis.
For bit vectors $x = (x_1,\ldots,x_n) \in \set{0,1}^n$ and  $v = (v_1,\ldots,v_n) \in \set{+,\times}^n$ we then write $\ket{x}_v = \ket{x_1}_{v_i} \kron \cdots \kron \ket{x_n}_{v_n}$ where $\ket{x_i}_{+} := \ket{x_i}$  and $\ket{x_i}_{\times} := H\ket{x_i}$.

Subsystem $X$ of a bipartite quantum system $XE$ is called {\em classical}, if the state of $XE$ is given by a density matrix of the form
$$
\rho_{XE} = \sum_{x \in \mathcal X} P_X(x) \proj{x} \otimes \rho_{E}^x \, ,
$$
where $\mathcal X$ is a finite set of cardinality $|{\mathcal X}| = \dim(\H_X)$, $P_X:{\mathcal X} \rightarrow [0,1]$ is a probability distribution, $\set{\ket{x}}_{x \in \mathcal X}$ is some fixed orthonormal basis of $\H_X$, and $\rho_E^x$ is a density matrix on $\H_E$ for every \mbox{$x \in \mathcal X$}. Such a state, called {\em hybrid} or {\em cq-} (for {\em c}lassical-{\em q}uantum) state, can equivalently be understood as consisting of a {\em random variable} $X$ with distribution $P_X$, taking on values in $\mathcal X$, and a system $E$ that is in state $\rho_E^x$ exactly when $X$ takes on the value $x$. This formalism naturally extends to two (or more) classical systems $X$, $Y$ etc. For any event $\E$ (defined by $P_{\E|X}(x) = \Pr[\E|X=x]$ for all $x$), we may write 
\[
\rho_{XE|\E} := \sum_x P_{X|\E} \outs{x} \kron \rho_E^x.
\]
If the state of $XE$ satisfies $\rho_{XE} = \rho_X \otimes \rho_E$, where $\rho_X = \tr_E(\rho_{XE}) = \sum_x P_X(x) \proj{x}$ and $\rho_E = \tr_X(\rho_{XE}) = \sum_x P_X(x) \rho_E^x$, then $X$ is {\em independent} of $E$, and thus no information on $X$ can be obtained from system~$E$. Moreover, if $\rho_{XE} = \frac{1}{|{\mathcal X}|} \I_X \otimes \rho_E$, where $\I_X$ denotes the identity on $\H_X$, then $X$ is {\em random-and-independent} of $E$. %This is what is aimed for in quantum cryptography, when $X$ represents a classical cryptographic key and $E$ the adversary's potential quantum information on~$X$. 
%It is not too hard to see that for two hybrid states $\rho_{XE}$ and $\rho_{XE'}$ with the same (distribution of) $X$, the trace distance between $\rho_{XE}$ and $\rho_{XE'}$ can be computed as $\dist(\rho_{XE},\rho_{XE'}) = \sum_x P_X(x) \dist(\rho_{E}^x,\rho_{E'}^x)$. 
We also want to be able to express that a random variable $X$ is (close) to being independent of a quantum system $E$ \emph{when given a random variable $Y$}. Formally, this is expressed by saying that $\rho_{XYE}$ equals $\rho_{X \leftrightarrow Y \leftrightarrow E}$, where
\[
\rho_{X \leftrightarrow Y \leftrightarrow E}:=\sum_{x,y} P_{XY}(x,y) \outs{x} \kron \outs{y} \kron \rho_E^y.
\]
This notion, called \emph{conditional independence}, for the quantum setting was introduced in \cite{DFSS07}. 

For a matrix $\rho$, the trace norm is defined as 
$\| \rho \|_1 := \tr \sqrt {\rho \rho^*}$, where $\rho^*$ denotes the Hermitian transpose of $\rho$.  
\begin{definition}
\label{def:tracedist}
The \emph{trace distance} between two density matrices $\rho,\sigma \in \sdm(\H)$ is defined as $\dist(\rho,\sigma) := \tfrac12 \| \rho - \sigma \|_1$.
\end{definition}

If two states $\rho$ and $\sigma$ are $\varepsilon$-close in trace distance, i.e. $\tfrac12 \| \rho - \sigma \|_1 \leq \varepsilon$, we use $\rho \approx_\varepsilon \sigma$ as shorthand. 
In case of classical states, the trace distance coincides with the statistical distance. Moreover, the trace distance between two states cannot increase when applying the same quantum operation (i.e., CPTP map) %\serge{Okay like that?} 
to both states. As a consequence, if $\rho \approx_\varepsilon \sigma$ then the states cannot be distinguished with statistical advantage better than $\varepsilon$. 
\begin{definition}\label{def:distuni}
For a density matrix $\rho_{XE}\in \sdm(\H_X \kron \H_E)$ with classical $X$, the \emph{distance to uniform} of $X$ given $E$ is defined as 
\[
\distuni(X|E) : = \tfrac12 \| \rho_{XE} - \rho_U \kron \rho_E\|_1, 
\]
where $\rho_U:= \frac{1}{\dim(\H_X)}\id_X$.
\end{definition}
%where $\rho_U$ is the fully mixed state on $\H_X$. 
%If also $E$ is classical, then $\distuni(X|E)$ simplifies to
%\begin{align*}
%\distuni(X|E) &= \tfrac12 \sum_{x,e} |P_{XE}(x,e) - P_U(x)P_E(e)|  \\
%&= \sum_e P_E(e) \, \tfrac12 \sum_x \big| P_{X|E}( x | e ) - P_U(x) \big|.
%\end{align*}

%%%%%%%%%%%%%%%%%%%%%%%%%%%%%%%%%%%%%%%%%%%%%%%%%%%%%%%%%%%%%%
\section{Min-Entropy and Privacy Amplification}
%%%%%%%%%%%%%%%%%%%%%%%%%%%%%%%%%%%%%%%%%%%%%%%%%%%%%%%%%%%%%%
% We make use of the notion of the {\em conditional min-entropy} $\hmin{\rho_{XE}|E}$ of a classical system (i.e.~random variable) $X$ conditioned on a quantum system $E$, as introduced by Renner~\cite{Renner05}. 
% If the state $\rho_{XE}$ of hybrid system $XE$ is clear from the context, we may write $\hmin{X|E}$ instead of $\hmin{\rho_{XE}|E}$.
We make use of Renner's notion of the {\em conditional min-entropy} $\hmin(\rho_{AB}|B)$ of a  system $A$ conditioned on another system $B$~\cite{Renner05}. 
%\niek{give def. of min-entropy for arbitrary states}
%Although the notion makes sense for arbitrary states, we restrict to hybrid states $\rho_{XE}$ with classical $X$. 
If the state $\rho_{AB}$ is clear from the context, we may write $\hmin(A|B)$ instead of $\hmin(\rho_{AB}|B)$.
The formal definition is given by $\hmin(\rho_{AB}|B):= \sup_{\sigma_B}\max\Set{h \in \R}{2^{-h} \cdot \id_A \kron \sigma_B - \rho_{AB} \geq 0}$ where the supremum is over all density matrices $\sigma_B$ on $\H_B$. If $\H_B$ is the trivial space $\cnum$, we obtain the unconditional min-entropy of $\rho_A$, denoted as $\hmin(\rho_A)$, which simplifies to $\hmin(\rho_A) = - \log \lambda_{\max}(\rho_A)$, where $\lambda_{\max}(\rho_A)$ is the largest eigenvalue of $\rho_A$.

%The {\em chain rule} guarantees that $\hmin(X|E) \geq \hmin(X) -\log \mathrm{rank}(\rho_E) \geq  \hmin(X) - \log \dim (\H_E)$ for classical $X$, where here and throughout this paper $\log$ denotes the {\em binary} logarithm. %, whereas $\ln$ denotes the {\em natural} logarithm. 

We will need the following chain rule.
%To prove \refthm{bqsm} we will use the following lemma.

\begin{lemma}
\label{lem:bqsmchain}
For any density matrix $\rho$ on $\H_{XYE}$ with classical $X$ and $Y$ it holds that
\[
\hmin(X|YE) \geq \hmin(X|Y) - \hmax(E).
\]
\end{lemma}
\noindent The proof can be found in \refapp{bqsmchainproof}.

For the special case of a hybrid state $\rho_{XE} \in \sdm(\H_X \kron \H_E)$ with classical $X$, it is shown in \cite{koenig09} that the conditional min-entropy of a quantum state coincides with the negative logarithm of the \emph{guessing probability conditional on quantum side information}
\[
\gs(X|E):= \max_{\set{M_x}}\sum_x P_X(x)\, \trace(M_x \rho_E^x),
\]
where the latter is the probability that the party holding $\H_E$ guesses $X$ correctly using the POVM $\set{M_x}_x$ on $\H_E$ that maximizes $\gs$.
Thus,
\begin{equation}
\label{eq:guessform}
\hmin(X|E) = -\log \gs(X|E).
\end{equation}
\noindent For random variables $X$ and $Y$, we have that $\gs(X|Y)$ simplifies to
\[
\gs(X|Y)=\sum_y P_Y(y) \gs(X|Y=y) = \sum_y P_Y(y)\max_x P_{X|Y}(x|y).
\] 

%The \minent{} of a classical random variable $X$ simplifies to the negative logarithm of guessing probability, i.e. $\hmin(X):= - \log \max_x P_X(x)$.

Finally, we make use of Renner's privacy amplification theorem~\cite{RK05,Renner05}, as given below. 
Recall that a function $g:\mathcal{R} \times \mathcal{X} \rightarrow \set{0,1}^\ell$ is called a {\em universal} (hash) function, if for the random variable $R$, uniformly distributed over $\mathcal{R}$, and for any distinct $x,y \in \mathcal{X}$: $\Pr[g(R,x)\!=\!g(R,y)] \leq 2^{-\ell}$.

\begin{thm}[Privacy amplification]\label{thm:PA}
Let $\rho_{XE}$ be a hybrid state with classical $X$. Let $g:\mathcal{R} \times \mathcal{X} \to \set{0,1}^\ell$ be a universal hash function, and let $R$ be uniformly distributed over $\mathcal{R}$, independent of $X$ and~$E$. Then $K = g(R,X)$ satisfies
%$$
%\dist\bigl( \rho_{KRE},{\textstyle \frac{1}{|\mathcal{K}|}} \id_K \otimes \rho_{RE} \bigr) \leq \frac12 \cdot 2^{-\frac12(\hmin(X|E) - \ell)} \, . 
%$$
$$
\distuni(K|RE) \leq \frac12 \cdot 2^{-\frac12(\hmin(X|E) - \ell)} \, . 
$$
\end{thm}

\noindent
Informally, Theorem~\ref{thm:PA} states that if $X$ contains sufficiently more than $\ell$ bits of entropy when given $E$, then $\ell$ nearly random-and-independent bits can be extracted from $X$.

\chapter{The All-But-One Entropic Uncertainty Relation} \label{sec:moreunbiasedbases}

Throughout this section, $\set{\B_1,\ldots,\B_m}$ is an arbitrary but fixed family of bases for the state space $\H$ of a quantum system. For simplicity, we restrict our attention to an $n$-qubit system, such that $\H = (\mathbb{C}^2)^{\otimes n}$ for $n \in \N$, but our results immediately generalize to arbitrary quantum systems. 
We write the $2^n$ basis vectors of the $j$-th basis $\B_j$ as $\B_j = \Set{\ket{x}_j}{x \in \set{0,1}^n}$. 
Let $c$ be the maximum overlap of $\set{\B_1,\ldots,\B_m}$, i.e., 
\[
c:= \max\Set{|\bra{x}_j \ket{y}_k|}{x,y \in \set{0,1}^n, 1 \!\leq\! j \!<\! k \!\leq\! m}.
\]

% Throughout this section, $n$ is an arbitrary but fixed positive integer, and $\set{\B_1,\ldots,\B_m}$ is an arbitrary family of bases for an $n$-qubit state with maximum overlap $c$, i.e., $c:= \max\Set{|\braket{\phi}{\psi}|}{\ket{\phi} \in \B_j, \ket{\psi} \in \B_k, 1 \!\leq\! j \!<\! k \!\leq\! m}$.
%
%$|\braket{\varepsilon}{\gamma}| \leq c$ for all basis vector pairs $\ket{\varepsilon}, \ket{\gamma}$ coming from different bases.

%and $\cC$ is an arbitrary but fixed binary (not necessary linear) code with minimal distance $d = \delta n$. It is convenient to view $\cC$ as a subset of $\set{+,\times}^n$. 

In order to obtain our entropic uncertainty relation that lower bounds the min-entropy of the measurement outcome for all but one measurement, we first show an uncertainty relation that expresses uncertainty by means of the probability measure of given sets. 
\begin{thm}[Theorem 4.18 in \cite{Schaffner07}] \label{thm:morehadamard}
  Let $\rho$ be an arbitrary state of $n$ qubits. For %$c \in \cC$, 
$j \in \setn[m]$, let $Q^j(\cdot)$ be the distribution of the outcome when
  $\rho$ is measured in the $\mathcal{B}_j$-basis, i.e., $Q^j(x) = \bra{x}_{j} \: \rho \: \ket{x}_{j}$ for any $x \in \set{0,1}^n$.%\footnote{Here and at other places, we write $\bra{x}_{j} \: \rho \: \ket{x}_{j}$, merely to avoid double subscripts; it should be understood as $\bra{x}_{\mathcal{B}_j} \: \rho \: \ket{x}_{\mathcal{B}_j}$.}
  Then, for any family $\set{\L^j}_{j \in \setn[m]}$ of subsets $\L^j \subset \set{0,1}^n$, it holds that
\[ \sum_{j \in \setn[m]} Q^j(\L^j) \leq 1 + c \, (m-1) \cdot  \max_{j \neq k \in \setn[m]} \sqrt{|\L^j|
   |\L^k|}. \]
\end{thm}
%
%\niek{====== I've made changes up to here in this Section ====}
A special case of Theorem~\ref{thm:morehadamard}, obtained by restricting the 
family of bases to 
%
%arbitrary code $\cC$ to the simple repetition code $\cC = 
% $\set{+\cdots+,\times\cdots\times}\subset \set{+,\times}^n$ 
%\carlos {Should we delete "$\subset \set{+,\times}^n$"?, I think it is not clear now why this appears here.} \niek{On the other hand, it gives additional information (i.e. the length $n$). Serge, how would you feel about this?}\serge{Changed it to the following; is maybe more consitent with the new notation.} 
the specific choice $\set{\B_+,\B_\times}$ with $\B_+ = \Set{\ket{x}}{x \in \set{0,1}^n}$ and $\B_\times = \Set{H^{\otimes n}\ket{x}}{x \in \set{0,1}^n}$
(i.e. either the computational or Hadamard basis for all qubits), is an uncertainty relation that was proven and used in the original paper about the BQSM~\cite{DFSS05}. 
The proof of Theorem~\ref{thm:morehadamard} goes along similar lines as the proof in the journal version of~\cite{DFSS05} for the special case outlined above. %of the repetition code. 
It is based on the norm inequality 
$$
\big\| A_1+\ldots+A_m \big\| \leq 1 + (m-1) \cdot \max_{j \neq k \in \setn[m]} 
\big\|A_j A_k\big\| \, ,
$$%\niek{should we specify that $i,k \in \setn[m]$?}
which holds for arbitrary orthogonal projectors $A_1,\ldots,A_m$. 
Recall that for a linear operator $A$ on the complex Hilbert space ${\mathcal{H}}$, the {\em operator norm} is defined as $\|A \| \assign \sup \|A \ket{\psi}\|$, 
where the supremum is over all norm-$1$ $\ket{\psi} \in \H$; this is identical to $\| A \| \assign \sup |\bra{\varphi}A\ket{\psi}|$, where the supremum is over all norm-$1$ $\ket{\varphi},\ket{\psi} \in \H$. 
Furthermore, $A$ is called an \emph{orthogonal projector} if $A^2 = A$ and $A^* = A$. The proof of this norm inequality can be found in Appendix~\ref{sec:proofinequality}. The proof of Theorem~\ref{thm:morehadamard} is given here. 
\begin{proof}[Proof of Theorem~\ref{thm:morehadamard}]
For $j \in \setn[m]$, we define the orthogonal projectors $A^j \assign \sum_{x \in \L^j} \ket{x}_{j} \bra{x}_{j}$. 
% \[
% A^j \assign \sum_{x \in \L^j} \ket{x}_{j} \bra{x}_{j} .
% \]
Using the spectral decomposition of $\rho = \sum_w \lambda_w
\proj{\varphi_w}$ and the linearity of the trace, we have
\begin{align*}
\sum_{j \in \setn[m]} Q^j(\L^j) &= \sum_{j \in \setn[m]} \tr(A^j\rho) 
 = \sum_{j \in \setn[m]} \sum_w
\lambda_w \tr(A^j \proj{\varphi_w})
= \sum_w \lambda_w \bigg( \sum_{j \in \setn[m]} \bra{\varphi_w}A^j\ket{\varphi_w}  \bigg)\\
&= \sum_w \lambda_w \bra{\varphi_w} \bigg( \sum_{j \in \setn[m]} A^j \bigg) \ket{\varphi_w}
\leq \bigg\| \sum_{j \in \setn[m]} A^j \bigg\| \leq 1 + (m-1) \cdot \max_{j \neq k \in\setn[m]} \big\|A^j A^k\big\|,
\end{align*}
where the last inequality is the norm inequality (Proposition~\ref{prop:morebases} in Appendix~\ref{sec:proofinequality}).
%% For any pure state $\ket{\varphi}$ measured, we have
%% \begin{equation} \label{eq:normbound}
%% \Qp(L^+)+\Qt(L^{\times}) = \bra{\varphi}A\ket{\varphi} +
%% \bra{\varphi}B\ket{\varphi} = \bra{\varphi}(A+B)\ket{\varphi} \leq \|A+B\|.
%% \end{equation}
%% As $\rho$ can be spectral-decomposed into $\rho = \sum_w \lambda_w
%% \proj{\varphi_w}$, Inequality \eqref{eq:normbound} also holds for
%% mixed states $\rho$ by convexity. 
To conclude, we show that $\|A^j A^k\| \leq c  \sqrt{|\L^j|
  |\L^k|}$. Let us fix $j \neq k \in \setn[m]$. Note that by the restriction on the overlap of the family of bases $\set{\mathcal{B}_j}_{j \in \setn[m]}$, we have that
%
%minimal distance of the code $\cC$, $c$ and $c'$ differ in at least $\delta n$ positions and therefore,  
$| \bra{x}_{j} \ket{y}_k | \leq c$ holds for all $x,y \in \set{0,1}^n$. %Further note that an arbitrary state $\ket{\psi} = \sum_z \lambda_z \ket{z}_k$ can be expressed with coordinates $\lambda_z$ in basis $c'$. 
Then, with the sums over $x$ and $y$ understood as over $x \in \L^j$ and $y \in \L^k$, respectively, 
\begin{align*}
\Big\| A^j A^k \ket{\psi} \Big\|^2 &= \bigg\| \sum_x \ket{x}_{j} \bra{x}_{j} \sum_y \ket{y}_k \bra{y}_k \ket{\psi} \bigg\|^2 
= \bigg\| \sum_x \ket{x}_{j} \sum_y \bra{x}_{j} \ket{y}_k \, \bra{y}_k \ket{\psi} \bigg\|^2  \\
&= \sum_x \bigg| \sum_y \bra{x}_{j} \ket{y}_k \, \bra{y}_k \ket{\psi} \bigg|^2 
\leq \sum_x \bigg(\sum_y \big|\bra{x}_{j} \ket{y}_k \, \bra{y}_k \ket{\psi} \big| \bigg)^2 \\
&\leq c^2 \sum_x \bigg(\sum_y \big|\bra{y}_k \ket{\psi}\big|\bigg)^2 \leq c^2  \big|\L^j\big| \big|\L^k\big| .
\end{align*}
The third equality follows from Pythagoras, the first inequality holds by triangle inequality, the second inequality by the bound on $| \bra{x}_{j} \ket{y}_k|$, and the last follows from Cauchy-Schwarz. 
This implies $\|A^j A^k\| \leq c \sqrt{|\L^j| |\L^k|}$ and finishes the proof. 
\end{proof}

In the same spirit as in (the journal version of)~\cite{DFSS05}, we reformulate above uncertainty relation in terms of a ``good event'' $\mathcal{E}$, which occurs with reasonable probability, and if it occurs, the measurement outcomes have high min-entropy. The statement is obtained by choosing the sets $\L^j$ in Theorem~\ref{thm:morehadamard} appropriately. %; the proof is given in Appendix~\ref{sec:proofcor}. 

Because we now switch to entropy notation, it will be 
%will now express results in terms of entropies, we find it more 
convenient to work with a measure of overlap between bases that is
logarithmic in nature and \emph{relative} to the number $n$ of qubits. Hence, we define
\[
\delta := - \frac{1}{n}\log c^2 \, . 
\]
We will later see that for ``good'' choices of bases, $\delta$ stays constant for growing $n$.

\begin{corollary} \label{cor:morehadamard}
Let $\rho$ be an arbitrary $n$-qubit state, let $J$ be a random variable over $\setn[m]$ (with arbitrary distribution~$P_J$), and let $X$ be the outcome when measuring $\rho$ in basis $\mathcal{B}_J$.%
\footnote{I.e., $P_{X\mid J}(x|j) = Q^j(x)$, using the notation from Theorem~\ref{thm:morehadamard}.} 
Then, for any $0< \epsilon< \delta/4 $, there exists %$\kappa > 0$ and an
an event $\mathcal{E}$ such that
$$
\sum_{j \in \setn[m]} \Pr[{\mathcal{E}} | J\!=\!j ] \geq (m-1) - (2m-1) \cdot 2^{-\epsilon n}  
$$
and 
%thus $\Pr[{\mathcal{E}}] \geq  (1-p) - p(2m-1) \cdot 2^{-\kappa n}$ where $p=\max_c P_C(c)$, and such that 
$$
\hmin(X | J\!=\!j,{\mathcal{E}}) \geq \Bigl(\frac{\delta}{2} - 2 \epsilon\Bigr) n  
$$
for $j \in \setn[m]$ with $P_{J\mid {\mathcal{E}}}(j) > 0$. 
\end{corollary}

\begin{proof}
For $j \in \setn[m]$ define
\begin{align*}
\mathcal S^j \assign \big\{ x \in \nbit &: Q^j(x) \leq  2^{-(\delta/2-\epsilon )n} \big\}
\end{align*} 
to be the sets of strings with small probabilities and denote by $\L^j
\assign \ol{\mathcal S}^j$ their complements\footnote{Here's the mnemonic: $\mathcal S$ for the strings with \emph{S}mall probabilities, $\L$ for \emph{L}arge.}. Note that for all $x \in \L^j$, we have that $Q^j(x) > 2^{-(\delta/2 - \epsilon )n}$ and therefore $|\L^j| < 2^{(\delta/2 - \epsilon)n}$. 
%For ease of notation, we
%abbreviate the probabilities that strings with small probabilities
%occur with $q^r \assign Q^r(S^r)$. 
It follows from Theorem~\ref{thm:morehadamard} that 
\begin{align*}
\sum_{j \in \setn[m]} Q^j(\mathcal S^j) &= \sum_{j \in \setn[m]} (1- Q^j(\L^j) ) \geq m - (1 + (m-1) \cdot 2^{-\epsilon n} )
=(m-1) - (m-1)2^{-\epsilon n}.
\end{align*}

We define ${\mathcal{E}} \assign \set{X \in \mathcal S^J \, \wedge \, Q^J(\mathcal S^J) \geq 2^{-\epsilon n}}$ to be the event that $X \in \mathcal S^J$ and at the same time the probability that this happens is not too small.
Then $\Pr[{\mathcal{E}}|J\!=\!j] = \Pr[X\in \mathcal S^j \wedge Q^j(\mathcal S^j)\geq 2^{-\epsilon n} |J\!=\!j]$ either vanishes (if $Q^j(\mathcal S^j) < 2^{-\epsilon n}$) or else equals $Q^j(\mathcal S^j)$. In either case, $\Pr[{\mathcal{E}}|J\!=\!j] \geq Q^j(\mathcal S^j) - 2^{-\epsilon n}$ holds and thus the first claim follows by summing over $j \in\setn[m]$ and using the derivation above. 
Furthermore,  let $p=\max_j P_J(j)$, then
$\Pr[ \bar{ \mathcal{E}}] = \sum_{j \in \setn[m]} P_J(j) \Pr[ \bar{\mathcal{E}}|J\!=\!j] \leq p \sum_{j \in \setn[m]} \Pr[\bar{\mathcal{E}}|J\!=\!j] \leq p (m-(\sum_{j\in \setn[m]} Q^j(\mathcal S^j) - 2^{-\epsilon n})) \leq p (1 +(2m-1) \cdot 2^{-\epsilon n})$, and $\Pr[{\mathcal{E}}] \geq  (1-p) - p(2m-1) \cdot 2^{-\epsilon n}$ 

Regarding the second claim, in case $J=j$, we have
\begin{align*}
  \hmin(X|J\!=\!j, {\mathcal{E}}) 
&= -\log\left(\max_{x \in \mathcal S^j} \frac{Q^j(x)}{Q^j(\mathcal S^j)}\right) %\nonumber
 \\& \geq -\log\left(\frac{2^{-(\delta/2 -\epsilon)n}}{Q^j(\mathcal S^j)}\right) = (\delta/2 -\epsilon) n + \log(Q^j(\mathcal S^j)). %\label{eq:Hinf}
\end{align*}
As $Q^j(\mathcal S^j) \geq 2^{-\epsilon n}$ by definition of ${\mathcal{E}}$, we have $\hmin(X|J\!=\!j, {\mathcal{E}}) \geq (\delta/2 - 2\epsilon) n$. 
\end{proof}

\section{Main Result and Its Proof}
We are now ready to state and prove our new all-but-one entropic uncertainty relation. 

\begin{thm}
\label{thm:UR}
%Let $\rho,\C, m$ and $\delta$ be as in \refthm{morehadamard}, and $R$ and $X$ as in \refcor{morehadamard}. 
Let $\rho$ be an arbitrary $n$-qubit state, let $J$ be a random variable over $\setn[m]$ (with arbitrary distribution~$P_J$), and let $X$ be the outcome when measuring $\rho$ in basis $\mathcal{B}_J$.
%random variable over $\mathcal{C}$, let $m:=|\mathcal C|$, let $\rho$ be any $n$-qubit state, and let $X$ be the outcome when $\rho$ is measured in basis $R$. 
Then, for any $0< \epsilon < \delta /4 $, there exists a random variable $J'$ with joint distribution $P_{JJ'X}$ such that (1) $J$ and $J'$ are % $(2m \cdot 2^{-\kappa n})$-
independent
%$ in statistical distance, i.e.,
%\[
%\tfrac{1}{2}\| P_{RR'} - P_{R} P_{R'} \|_1 < 2 m\cdot 2^{-\kappa n}.
%\]
and (2) there exists an event $\Psi$ with $\Pr[\Psi] \geq 1-2\cdot2^{-\epsilon n} $ such that%
\footnote{Instead of introducing such an event $\Psi$, we could also express the min-entropy bound by means of the {\em smooth} min-entropy of $X$ given $J=j$ and $J'=j'$. }
\[
\hmin(X|J=j,J'=j',\Psi) \geq \Bigl(\frac{\delta}{2} - 2 \epsilon\Bigr) n  - 1 
\]
for all $j,j' \in \setn[m]$ with $j \neq j'$ and $\P_{JJ'|\Psi}(j,j')>0$. 
%Furthermore, the random variable $J'$ is the same for any choice of $C$
\end{thm}

Note that, as phrased, Theorem~\ref{thm:UR} requires that $J$ is fixed and known, and only then the existence of $J'$ can be guaranteed. This is actually not necessary. By looking at the proof, we see that $J'$ can be defined simultaneously in all $m$ probability spaces $P_{X|J=j}$ with $j \in \setn[m]$, without having assigned a probability distribution to $J$ yet, so that the resulting random variable $J'$ we obtain by assigning an {\em arbitrary} probability distribution $P_J$ to $J$, satisfies the claimed properties. This in particular implies that the (marginal) distribution of $J'$ is fully determined by $\rho$. 

%We also would like to point out that Theorem~\ref{thm:UR} generalizes to an {\em arbitrary} set of $m$ measurements, not necessarily characterized by a code $\cC$, as long as the overlap between any two bases is at most $2^{-\delta n/2}$, i.e., $|\braket{\varepsilon}{\gamma}| \leq 2^{-d/2}$ for all basis vector pairs $\ket{\varepsilon}, \ket{\gamma}$ coming from different bases. 

The idea of the proof of Theorem~\ref{thm:UR} is to (try to) define the random variable $J'$ in such a way that the event $J \neq J'$ coincides with the ``good event'' $\mathcal{E}$ from Corollary~\ref{cor:morehadamard}. It then follows immediately from Corollary~\ref{cor:morehadamard} that $\hmin(X|J=j,J' \neq J) \geq (\delta/2-2\epsilon) n$, which is already close to the actual min-entropy bound we need to prove. This approach dictates that if the event $\mathcal{E}$ does not occur, then $J'$ needs to {\em coincide} with~$J$. Vice versa, if $\mathcal{E}$ does occur, then $J'$ needs to be {\em different} to $J$. However, it is a priori unclear {\em how} to choose $J'$ different to $J$ in case $\mathcal{E}$ occurs. 
There is only one way to set $J'$ to be equal to $J$, but there are many ways to set $J'$ to be different to $J$ (unless $m = 2$).
It needs to be done in such a way that without conditioning on $\mathcal{E}$ or its complement, $J$ and $J'$ are independent.

Somewhat surprisingly, it turns out that the following does the job. To simplify this informal discussion, we assume that the sum of the $m$ probabilities $\Pr[{\mathcal{E}} | J\!=\!j ]$ from Corollary~\ref{cor:morehadamard} equals $m-1$ exactly. It then follows that the corresponding complementary probabilities, $\Pr[\bar{\mathcal{E}} | J\!=\!j ]$ for the $m$ different choices of $j \in \setn[m]$, add up to $1$ and thus form a probability distribution. $J'$ is now chosen, in the above spirit depending on the event $\mathcal{E}$, so that its marginal distribution $P_{J'}$ coincides with this probability distribution: $P_{J'}(j') = \Pr[\bar{\mathcal{E}} | J\!=\!j' ]$ for all $j'\in \setn[m]$. Thus, in case the event $\mathcal{E}$ occurs, $J'$ is chosen according to this distribution but conditioned on being different to the value $j$, taken on by $J$. 
The technical details, and how to massage the argument in case the sum of the $\Pr[{\mathcal{E}} | J\!=\!j ]$'s is not exactly $m-1$, are worked out in the proof below.

\begin{proof}[Proof of Theorem~\ref{thm:UR}]
From \refcor{morehadamard} we know that for any $0<\epsilon<\delta/4$, there exists an event $\E$ such that $\sum_{j \in \setn[m]} \Pr[\E | J=j] =  m-1- \alpha$, and thus $\sum_{j \in \setn[m]} \Pr[\bar{\E} | J=j] = 1 + \alpha$, for $-1 \leq \alpha \leq (2m-1) 2^{-\epsilon n}$. We make a case distinction between $\alpha=0$, $\alpha>0$ and $\alpha <0$; we start with the case $\alpha = 0$, we subsequently prove the other two cases by reducing them to the case $\alpha = 0$ by ``inflating'' and ``deflating'' the event $\E$ appropriately. 
The approach for the case $\alpha = 0$ is to define $J'$ in such way that $\E \iff J \neq J'$, i.e., the event $ J \neq J'$ coincides with the event $\E$. The min-entropy bound from \refcor{morehadamard} then immediately translates to $\hmin(X | J=j, J' \neq J) \geq (\delta/2-2\epsilon) n$, and to $\hmin(X|J=j, J'=j') \geq (\delta/2-2\epsilon) n$ for $j' \neq j$ with $P_{JJ'}(j,j')>0$, as we will show. 
What is not obvious about the approach is how to define $J'$ when it is supposed to be different from $J$, i.e., when the event $\E$ occurs, so that in the end $J$ and $J'$ are independent. 

Formally, we define $J'$ by means of the following conditional probability distributions:  
%because, as will become clear from the proof, 
% This definition implies that 
\[
P_{J' | J X \bar \E}(j'|j,x) := %P_{R' | R  \bar \E}(r'|r)=
\left\{ \begin{array}{cl} 1 & \text{if } j=j' \\%\wedge \Pr[\E|R=r]>0 , \\ 
0 &  \text{if } j \neq j'  %\wedge \Pr[\E|R=r]>0 ,\\
%\Pr[\bar\E|R=r'] &  \text{if } \Pr[\E|R=r]=0.
\end{array}\right.  
\quad
\text{and}
\quad
P_{J'|J X \E}(j'|j,x) := %P_{R'|R \E}(r'|r) = 
\left\{ \begin{array}{cc} 0 & \text{if } j=j' \\ \displaystyle 
\frac{\Pr[\bar \E | J=j']}{\Pr[\E|J=j]}
 &  \text{if } j \neq j'
\end{array}\right. 
\]
We assume for the moment that the denominator in the latter expression does not vanish for any $j$; we take care of the case where it does later. 
Trivially, $P_{J'|JX\bar\E}$ is a proper distribution, with non-negative probabilities that add up to $1$, and the same holds for $P_{J'|JX\E}$:
\begin{align*}
\sum_{j' \in \setn[m]} P_{J'|JX\bar \E} = \sum_{j' \in \setn[m] \setminus \set{j}} P_{J'|JX\bar \E} &= \sum_{j' \in \setn[m] \setminus \set{j}}  \frac{\Pr[\bar \E | J=j']}{\Pr[\E |J=j]} 
%=  \sum_{r' \in \C\setminus \set{r}}  \frac{1 - \Pr[\E | R=r']}{\Pr[\E |R=r]} %\\
%&=   \frac{(m-1) - (m-1 -\alpha- \Pr[\E | R=r])}{\Pr[ \E | R=r] } = 1. \qquad \text{(since $\alpha=0$)}
= 1
\end{align*}
where we used that $\sum_{j \in \setn[m]} \Pr[\bar{\E} | J=j] =  1$ (because $\alpha = 0$) in the last equality. 
Furthermore, it follows immediately from the definition of $J'$ that $\bar \E \implies J=J'$ and $\E \implies J \neq J'$. Hence, $\E \iff J \neq J'$, and thus the bound from \refcor{morehadamard} translates to $\hmin(X | J=j, J' \neq J) \geq (\delta/2-2\epsilon) n$. It remains to argue that $J'$ is independent of $J$, and that the bound also holds for $\hmin(X|J=j, J'=j')$ whenever $j \neq j'$. 
%The bound on the min-entropy can now be argued as above. 

The latter follows immediately from the fact that conditioned on $J \neq J'$ (which is equivalent to $\E$), $X,J$ and $J'$ form a Markov chain $X \leftrightarrow J \leftrightarrow J'$, and thus, given $J=j$, additionally conditioning on $J'=j'$ does not change the distribution of $X$. 
%This property allows us to replace the event $R\neq R'$ by the event $R'=r'$ in the conditional min-entropy $\hmin(X|R=r,R \neq R')$ for any $c' \neq c$ with $P_{RR'}(r,r') > 0$.
For the independence of $J$ and $J'$, consider the joint probability distribution of $J$ and $J'$, given by
\begin{align*}
P_{JJ'} (j,j') &= P_{J'J\E}(j',j) + P_{J'J\bar\E}(j',j)\\
&= P_J(j)  \Pr[ \E | J=j] P_{J' | J \E}(j'|j)  + P_J(j)  \Pr[\bar \E | J=j] P_{J' | J \bar \E}(j' | j)
\\&= P_J(j) \Pr[\bar \E|J=j'],
\end{align*}
where the last equality follows by separately analyzing the cases $j=j'$ and $j \neq j'$. %I.e., in both cases, plugging in the the definitions for $P_{R' | R \E}(r' | r)$ and $P_{R' | R \bar \E}(r' | r)$ yields the same expression.
It follows immediately that the marginal distribution of $J'$ is $P_{J'}(j') = \sum_{j } P_{JJ'}(j,j') = \Pr[\bar \E|J=j']$, and thus $P_{JJ'} = P_J \cdot P_{J'}$.

What is left to do for the case $\alpha = 0$ is to deal with the case where there exists $j^*$ with $\Pr[\E | J=j^*]=0$. 
Since $\sum_{j \in \setn[m]} \Pr[\bar \E | J=j] = 1$, 
%Since $\Pr[\E | R=r^*]=0$, we have that $\Pr[\bar \E | R=r^*]=1$, which implies that 
it holds that $\Pr[\bar \E | J=j] =0 $ for $j\neq j^*$. This motivates to define $J'$ as $J' := j^*$ with probability $1$.
Note that this definition directly implies that $J'$ is independent from $J$. Furthermore, by the above observations: $\E \iff J \neq J'$. 
%\carlos{Note that in both cases the definition of the random variable $J'$ does not depend on the choice of $J$.} 
This concludes the case $\alpha=0$. %; the rest of the proof is found in \refapp{remainder}.

%What remains to prove are the cases where $\alpha \neq 0$.
%Next, we consider 
Next, we consider the case $\alpha > 0$. The idea is to ``inflate'' the event $\E$ so that $\alpha$ becomes $0$, i.e., to define an event $\E'$ that contains $\E$ (meaning that $\E \implies \E'$) so that $\sum_{j \in \setn[m]} \Pr[\E'|J=j] = m-1$, and to define $J'$ as in the case $\alpha = 0$ (but now using $\E'$). Formally, we define $\E'$ as the disjoint union $\E' = \E \vee \E_\circ$ of $\E$ and an event $\E_\circ$. The event $\E_\circ$ is defined by means of $\Pr[\E_\circ|\E, J=j,X=x] = 0$, so that $\E$ and $\E_\circ$ are indeed disjoint, and $\Pr[\E_\circ|J=j,X=x] = \alpha/m$, so that indeed 
$$
\sum_{j \in \setn[m]} \Pr[\E'|J=j] = \sum_{j \in \setn[m]} (\Pr[\E|J=j]+\Pr[\E_\circ|J=j]) = (m-1-\alpha)+\alpha = m-1 \, .
$$

We can now apply the analysis of the case $\alpha = 0$ to conclude the existence of $J'$, independent of $J$, such that $J \neq J' \iff \E'$ and thus $(J \neq J') \wedge \bar{\E}_\circ \iff \E'\wedge \bar{\E}_\circ \iff \E$. Setting $\Psi := \bar{\E}_\circ$, it follows that
\[
\hmin(X|J=j,J \neq J',\Psi) = \hmin(X|J=j,\E) \geq (\delta/2-2\epsilon) n \, ,
\]
where $\Pr[\Psi] = 1 - \Pr[\E_\circ] = 1 - \alpha/m \geq 1 - (2m-1) 2^{-\epsilon n}/m \geq 1 - 2 \cdot 2^{-\epsilon n}$. Finally, using similar reasoning as in the case $\alpha = 0$, it follows that the same bound holds for $\hmin(X|J=j, J'=j',\Psi)$ whenever $j \neq j'$. This concludes the case $\alpha > 0$. 

Finally, we consider the case $\alpha < 0$. The approach is the same as above, but now $\E'$ is obtained by ``deflating'' $\E$. 
Specifically, we define $\E'$ by means of $\Pr[\E'|\bar\E,J=j,X=x ]=\Pr[\E'|\bar\E ] =0$, so that $\E'$ is contained in $\E$, and $\Pr[\E' | \E, J=j, X=x] = \Pr[\E' | \E] = \frac{m-1}{m-1-\alpha}$, so that
$$
\sum_{j \in \setn[m]} \Pr[\E'|J=j] = \sum_{j \in \setn[m]} \Pr[\E'| \E]\cdot \Pr[ \E |J=j]  = m - 1 \, .
$$
Again, from the $\alpha = 0$ case we obtain $J'$, independent of $J$, such that the event $J \neq J'$ is equivalent to the event $\E'$.

It follows that 
\begin{align*}
\hmin(&X|J=j,J \neq J') = \hmin(X|J=j,\E')  = \hmin(X|J=j,\E',\E)  \\
&\geq \hmin(X|J=j,\E) - \log(P[\E'|\E,J=j]) \geq (\delta/2-2\epsilon) n - 1\, ,
\end{align*}
where the second equality holds because $\E' \implies \E$, the first inequality holds because additionally conditioning on $\E'$ increases the probabilities of $X$ conditioned on $J=j$ and $\E$ by at most a factor $1/P[\E'|\E,J=j])$, and the last inequality holds by \refcor{morehadamard}) and because $P[\E'|\E,J=j]) = \frac{m-1}{m-1-\alpha} \geq \frac12$, where the latter holds since $\alpha \geq -1$. 
Finally, using similar reasoning as in the previous cases, it follows that the same bound holds for $\hmin(X|J=j, J'=j')$ whenever $j \neq j'$. This concludes the proof. 
\end{proof}

\section{Constructing Good Families of Bases}
\label{sec:goodfam}
Here, we discuss some interesting choices for the family $\set {\B_1,\ldots,\B_m}$ of bases. We say that such a family is ``good'' if $\delta = - \frac{1}{n}\log(c^2)$ converges to a strictly positive constant as $n$ tends to infinity. There are various ways to construct such families. For example, a family obtained through sampling according to the Haar measure will be good with overwhelming probability (a precise statement, in which ``good'' means $\delta=0.9$, can be found at the very end of the proof of Theorem 2.5 of \cite{FHS11}). The best possible constant $\delta = 1$ is achieved for a family of \emph{mutually unbiased bases}. However, for arbitrary quantum systems (i.e., not necessarily multi-qubit systems) it is not well understood how large such a family may be, beyond that its size cannot exceed the dimension plus $1$. 

In the upcoming section, we will use the following simple and well-known construction. For an arbitrary binary code $\Tset{C} \subset \set{+,\times}^n$ of size $m$, minimum distance $d$ and encoding function $\encoding:\setn[m]\rightarrow \C$, we can construct a family $\set{\B_1,\ldots,\B_m}$ of bases as follows. We identify the $j$th codeword, i.e. $\encoding(j) =  (c_1,\ldots,c_n)$ for $j\in \setn[m]$, with the basis $\B_j  = \Set{\ket{x}_{\encoding(j)}}{x \in \set{0,1}^n}= \Set{ (H^{c_1} \!\kron \cdots \kron\! H^{c_n}) \ket{x}}{x \in \set{0,1}^n}$. In other words, $\B_j$ measures qubit-wise in the computational or the Hadamard basis, depending on the corresponding coordinate of $\encoding(j)$.  
It is easy to see that the maximum overlap $c$ of the family obtained this way is directly related to the minimum distance of $\C$, namely $\delta = -\frac 1 n \log(c^2)$ coincides with the relative minimal distance $d/n$ of $\C$. Hence, choosing an asymptotically good code immediately yields a good family of bases.

\chapter{Application: A New Quantum Identification Scheme} \label{sec:identification}

Our main application of the new uncertainty relation is in proving security of a new  identification scheme in the quantum setting.
The goal of (password-based) identification is to ``prove'' knowledge of a password $w$ (or some other low-entropy key, like a PIN) without giving $w$ away. 
More formally, given a user \user and a server \server that hold a pre-agreed password $w \in \Tset{W}$, \user wants to convince \server that he indeed knows $w$, but in such a way that he gives away as little information on $w$ as possible in case he is actually interacting with a dishonest server \dishserver. 

In~\cite{DFSS07}, Damg{\aa}rd \etal\ showed the existence of a secure identification scheme in the {\em bounded-quantum-storage} model. The scheme involves the communication of qubits, and is secure against an arbitrary dishonest server \server that has limited quantum storage capabilities and can only store a certain fraction of the communicated qubits, whereas the security against a dishonest user \dishuser holds unconditionally. 

On the negative side, it is known that {\em without} any restriction on (one of) the dishonest participants, secure identification is impossible (even in the quantum setting). Indeed, if a quantum scheme is unconditionally secure against a dishonest user, then unavoidably it can be broken by a dishonest server with unbounded quantum-storage and unbounded quantum-computing power; this follows essentially from \cite{lo96} (see also~\cite{DFSS07}). Thus, the best one can hope for (for a scheme that is unconditionally secure against a dishonest user) is that in order to break it, unbounded quantum storage {\em and} unbounded quantum-computing power is {\em necessary} for the dishonest server. 
This is not the case for the scheme of~\cite{DFSS07}: storing all the communicated qubits as they are, and measuring them qubit-wise in one or the other basis at the end, completely breaks the scheme. Thus, no quantum computing power at all is necessary to break the scheme, only sufficient quantum storage. 

In this section, we propose a new identification scheme, which can be regarded as a first step towards closing the above gap. Like the scheme from \cite{DFSS07}, our new scheme is secure against an unbounded dishonest user and against a dishonest server with limited quantum storage capabilities. The new uncertainty relation forms the main ingredient 
in the user-security proof in the BQSM.
% for proving security against a dishonest server with bounded quantum storage.
Furthermore, and in contrast to~\cite{DFSS07}, a minimal amount of quantum computation power is {\em necessary} to break the scheme, beyond sufficient quantum storage. 
Indeed, next to the security against a dishonest server with bounded quantum storage, we also prove---in \refsec{usecsqom}---security against a dishonest server that can store all the communicated qubits, but is restricted to measure them qubit-wise (in arbitrary qubit bases) at the end of the protocol execution. Thus, beyond sufficient quantum storage, quantum computation that involves {\em pairs} of qubits is necessary (and in fact sufficient) to break the new scheme.

Restricting the dishonest server to qubit-wise measurements may look restrictive; however, we stress that in order to break the scheme, the dishonest server needs to store many qubits {\em and} perform quantum operations on them that go beyond single-qubit operations; this may indeed be considerably more challenging than storing many qubits and measuring them qubit-wise. Furthermore, it turns out that proving security against such a dishonest server that is restricted to qubit-wise measurements is already challenging; indeed, standard techniques do not seem applicable here. Therefore, handling a dishonest server that can, say, act on {\em blocks} of qubits, must be left to future research. 

%The new uncertainty relation is the main ingredient  for proving security against a dishonest server with bounded quantum storage. Proving security against a dishonest server with unbounded quantum storage but is restricted to qubit-wise measurements, and proving security against an unbounded dishonest user, require different and new techniques. 

\section{Security Definitions}

We first formalize the security properties we want to achieve. We borrow the definitions from \cite{DFSS07}, which are argued to be ``the right ones'' in \cite{FS09}. 

\begin{definition}[Correctness]
\label{def:correctness}
An identification protocol is said to be \emph{$\varepsilon$-correct} if, after an execution by honest \user and honest \server, \server accepts with probability $1-\varepsilon$.%When protocol \QID is executed by honest \user and honest \server, then \server accepts with probability $1-\negl(n)$.
\end{definition}

\begin{definition}[User security]
\label{def:usec}
An identification protocol for two parties \user, \server is $\varepsilon$-secure for the user \user %with error $\varepsilon$ 
against (dishonest) server \dishserver if the following holds: If the initial state of \dishserver is independent of $W$, then its state $E$ after execution of the protocol is such that there exists a random variable $W^\prime$ that is independent of $W$ and such that 
\[
\rho_{WW^\prime E|W\neq W^\prime} \approx_\varepsilon \rho_{W \leftrightarrow W^\prime \leftrightarrow E | W\neq W^\prime}.
\]
\end{definition}

\begin{definition}[Server security]\label{def:servsec}
An identification protocol for two parties \user, \server is $\varepsilon$-secure for the server \server %with error $\varepsilon$ 
against (dishonest) user \dishuser if the following holds: whenever the initial state of \dishuser is independent of $W$, then there exists a random variable $W^\prime$ (possibly $\perp$) that is independent of $W$ such that if $W\neq W^\prime$ then \server 
accepts with probability at most $\varepsilon$. 
%and if $W=W^\prime$ then \server accepts with certainty
Furthermore, the common state $\rho_{W E}$ after execution of the protocol (including \server's announcement to accept or reject) satisfies
\[
\rho_{WW^\prime E|W\neq W^\prime} \approx_\varepsilon \rho_{W \leftrightarrow W^\prime \leftrightarrow E | W\neq W^\prime}.
\]
\end{definition}

We will prove the user-security of the protocol in two different models, in which different assumptions are made. Because these assumptions are in some sense ``orthogonal'', the hope is that if security would break down in one model to a failing assumption, the protocol is still secure by the other model.

\section{Description of the New Quantum Identification Scheme}
\label{sec:schemedescr}
Let $\Tset{C} \subset \set{+,\times}^n$ be a binary code with minimum distance $d$, and let $\encoding: \Tset{W} \rightarrow \Tset{C}$ be its %an injection, which we call the 
encoding function. Let $m:=|\Tset{W}|$, and typically, $m < 2^n$. Let $\mathcal{F}$ be the class of all linear functions from $\set{0,1}^n$ to $\set{0,1}^\ell$, where $\ell < n$, represented as $\ell \times n$ %full-rank 
binary matrices. It is well-known that this class is two-universal. Furthermore, let $\mathcal{G}$ be a strongly two-universal class of hash functions from $\Tset{W} $ to $\{0,1\}^\ell$. Protocol \QID is shown below.
\begin{protocol}
\begin{enumerate*}
%\item \user picks $x\in_R \{0,1\}^n$ and sends $\ket{x}_{\encoding(w)}$ to \server.
\item \user picks $x\in \{0,1\}^n$ independently and uniformly at random and sends $\ket{x}_{\encoding(w)}$ to \server.
\item \server measures in basis $\encoding(w)$. Let $x^\prime$ be the outcome.
\item \user picks $f \in \mathcal{F}$ independently and uniformly at random and sends it to \server
\item \server picks $g\in \mathcal{G}$ independently and uniformly at random and sends it to \user
\item \user computes and sends $z:=f(x)\oplus g(w) $ to \server
\item  \server accepts if and only if $z=z^\prime$ where $z^\prime :=f(x') \oplus g(w)$
\end{enumerate*}
\caption{\QID}
\end{protocol}

Our scheme is quite similar to the scheme in~\cite{DFSS07}. The difference is that in our scheme, both parties, \user and \server, use $\encoding(w)$ as basis for preparing/measuring the qubits in step (1) and (2), whereas in~\cite{DFSS07}, only \server uses $\encoding(w)$ and \user uses a {\em random} basis $\theta \in \set{+,\times}^n$ instead, and then \user communicates $\theta$ to \server and all the positions where $\theta$ and $\encoding(w)$ differ are dismissed. 
Thus, in some sense, our new scheme is more natural since why should \user use a random basis when he knows the right basis (i.e., the one that \server uses)? 
In~\cite{DFSS07}, using a random basis (for \user) was crucial for their proof technique, which is based on an entropic uncertainty relation of a certain form, which asks for a random basis. 
However, using a random basis, which then needs to be announced, renders the scheme insecure against a dishonest server \dishserver that is capable of storing all the communicated qubits and then measure them in the right basis once it has been announced. 
Our new uncertainty relation applies to the case where an $n$-qubit state is measured in a basis that is sampled from a code $\Tset{C}$, and thus is applicable to the new scheme where \user uses basis $\encoding(w) \in \Tset{C}$. Since this basis is common knowledge (to the honest participants), it does not have to be communicated, and as such a straightforward store-and-then-measure attack as above does not apply. 

A downside of our scheme is that security only holds in case of a perfect quantum source, which emits exactly one qubit when triggered. Indeed, a multi-photon emission enables a dishonest server \dishserver to learn information on the basis used, and thus gives away information on the password $w$ in our scheme.  As such, our scheme is currently mainly of theoretical interest. 

It is straightforward to verify that (in the ideal setting with perfect sources, no noise, etc.) \QID satisfies the correctness property (Definition~\ref{def:correctness}) perfectly, i.e.\ $\varepsilon=0$. In the remaining sections, we prove (unconditional) security against a dishonest user, and we prove security against two kinds of restricted dishonest servers. First, against a dishonest server that has limited quantum storage capabilities, and then against a dishonest server that can store an unbounded number of qubits, but can only store and measure them qubit-wise. 

\section{(Unconditional) Server Security}
First, we claim security of \QID against an arbitrary dishonest user \dishuser (that is merely restricted by the laws of quantum mechanics).
\label{sec:serversec}
\begin{thm}
\label{thm:serversec}
\QID is $\varepsilon$-secure for the server with $\varepsilon=\binom{m}{2}2^{-\ell}$.
\end{thm}

\begin{proof}
Clearly, from the steps (1) to (5) in the protocol \QID, \dishuser learns no information on $W$ at all. The only information he may learn is by observing whether \server accepts or not in step (6). Therefore, in order to prove server security, it suffices to show the existence of a random variable $W'$, independent of $W$, with the property that \server rejects whenever $W' \neq W$ (except with probability $\frac12 m(m-1) 2^{-\ell}$). 

We may assume that $\mathcal{W} = \set{1,\ldots,m}$. 
Let $\rho_{WX'FGZE}$ be the state describing the password $W$, the variables $X',F,G$ and $Z$ occurring in the protocol from the server's point of view, and \dishuser's quantum state $E$ {\em before} observing \server's decision to accept or reject. For any $w \in \mathcal{W}$, consider the state $\rho^w_{X'FGZE} := \rho_{X'FGZE|W=w}$. Note that the reduced state $\rho^w_{FGZE}$ is the same for any $w \in \mathcal{W}$; this follows from the assumption that \dishuser's initial state is independent of $W$ and because $F,G$ and $Z$ are produced independently of $W$. We may thus write $\rho^w_{X'FGZE}$ as $\rho_{X'_w F G Z E}$, and we can ``glue together'' the states $\rho_{X'_w F G Z E}$ for all choices of $w$. This means, there exists a state $\rho_{X'_1 \cdots X'_m F G Z E_1 \cdots E_m}$ that correctly reduces to $\rho_{X'_w F G Z E_w} = \rho_{X'_w F G Z E}$ for any $w \in \mathcal{W}$, and conditioned on $FGZ$, we have that $X'_i E_i$ is independent of $X'_j E_j$ for any $i \neq j \in \mathcal{W}$. 
It is easy to see that for any $i \neq j \in \mathcal{W}$, $G$ is independent of $X'_i,X'_j$ and $F$. Therefore, by the strong two-universality of $G$, for any $i \neq j$ it holds that $Z'_i \neq Z'_j$ except with probability $2^{-\ell}$, where $Z'_w = FX'_w+G(w)$ for any $w$. Therefore, by the union bound, $Z'_1,\ldots,Z'_m$ are pairwise distinct and thus $Z$ can coincide with at most one of the $Z'_w$'s, except with probability $\varepsilon = \frac12 m(m-1) 2^{-\ell}$. Let $W'$ be defined such that $Z = Z'_{W'}$; if there is no such $Z'_w$ then we let $W' = \: \perp$, and if there are more than one then we let it be the first. Recall, the latter can happen with probability at most $\varepsilon$. 
We now extend the state $\rho_{X'_1 \cdots X'_m F G Z W' E_1 \cdots E_m}$ by $W$, chosen independently according to $P_W$. Clearly $W'$ is independent of $W$. Furthermore, except with probability at most $\varepsilon$, if $W \neq W'$ then $Z \neq Z'_{W}$. Finally note that $\rho_{X'_W F G Z W' W E_W}$ is such that $\rho_{X'_W F G Z W E_W} = \sum_w P_W(w)\rho_{X'_w F G Z E_w} \otimes \ketbra{w}{w} = \sum_w P_W(w)\rho^w_{X' F G Z E} \otimes \ketbra{w}{w} = \rho_{X' F G Z W E}$. Thus, also with respect to the state $\rho_{X' F G Z W E}$ there exist $W'$, independent of $W$, such that if $W' \neq W$ then $Z \neq Z'$ except with probability at most $\varepsilon$. 
This was to be shown. 
\end{proof}

\section{User Security in the Bounded-Quantum-Storage Model}
\label{sec:bqsm}

Next, we consider a dishonest server \dishserver, and first prove security of \QID in the {\em bounded-quantum-storage model}. In this model, as introduced in~\cite{DFSS05}, it is assumed that the adversary (here \dishserver) cannot store more than a fixed number of qubits, say~$q$. 
The security proof of \QID in the bounded-quantum-storage model is very similar to the corresponding proof in \cite{DFSS07} for their scheme, except that we use the new uncertainty relation from \refsec{moreunbiasedbases}. Furthermore, since our uncertainty relation (\refthm{UR}) already guarantees the existence of the random variable $W'$ as required by the security property, no {\em entropy-splitting} as in~\cite{DFSS07} is needed. 

In the following, let $\delta:=d/n$, i.e. the relative minimum distance of $\C$.
\begin{thm}\label{thm:bqsm} 
Let \dishserver be a dishonest server
% the initial state of a dishonest server \dishserver, 
whose quantum memory is at most $q$ qubits at step 3 of \QID. 
%be independent of the user's key $W$.
Then, for any $0<\kappa< \delta / 4$, \QID is $\varepsilon$-secure for the user with 
\[
\varepsilon = 2^{-\frac 1 2 ((\delta/2-2\kappa) n -1 - q - \ell)}+ 4\cdot 2^{-\kappa n}.
\]
\end{thm}
\begin{proof}
We consider and analyze a purified version of \QID where in step (1) instead of sending $\ket{X}_c$ to \dishserver for a uniformly distributed $X$, \user prepares a fully entangled state $2^{-n/2}\sum_x \ket{x}\ket{x}$ and sends the second register to \dishserver while keeping the first. Then, in step (3) when the memory bound has applied, \user measures his register in the basis $\encoding(W)$ in order to obtain $X$. Note that this procedure produces exactly the same common state as in the original (non-purified) version of \QID. Thus, we may just as well analyze this purified version. 

The state of \dishserver consists of his initial state and his part of the EPR pairs, and may include an additional ancilla register.
Before the memory bound applies, \dishserver may perform any unitary transformation on his composite system. 
When the memory bound is applied (just before step (3) is executed in \QID), \dishserver has to measure all but $q$ qubits of his system. Let the classical outcome of this measurement be denoted by $y$, and let $E'$ be the %cq-state consisting $Y$ and the 
remaining quantum state of at most $q$ qubits. The common state has collapsed to a $(n+q)$-qubit state and depends on $y$; the analysis below holds for any $y$.
Next, \user measures his $n$-qubit part of the common state in basis $\encoding(W)$; let $X$ denote the classical outcome of this measurement. 
By our new uncertainty relation (\refthm{UR}) and subsequently applying the min-entropy chain rule that is given in \reflem{bqsmchain} (to take the $q$ stored qubits into account) it follows that there exists $W'$, independent of $W$, and an event $\Psi$ that occurs at least with probability $1-2\cdot 2^{-\kappa n}$, such that 
%conditioned on $\Psi$ and on $W \neq W'$, the measurement outcome of \user, $X$, has at least $(\delta/2-2\kappa) n -1-q$ bits of min-entropy conditioned on the view of \dishserver. 
\[
\hmin(X|E',W=w, W'=w', \Psi ) \geq (\delta/2-2\kappa) n -1-q.
\]
for any $w,w'$ such that $w\neq w'$.
Because \user chooses $F$ independently at random from a 2-universal family, privacy amplification guarantees that 
\[
\distuni(F(X)|E'F,W=w,W'=w') \leq \varepsilon' := \frac 1 2 \cdot 2^{-\frac 1 2 ((\delta/2-2\kappa) n -1 -q-\ell)} + 2\cdot 2^{-\kappa n},
\]
for any $w,w'$ such that $w\neq w'$.
%$F(X)$ is $\varepsilon'$-close to random-and-independent of $F$, $W$, $W'$ and $E'$, conditioned on $W \neq W'$ (but \emph{not} conditioned on the event $\Psi$ anymore, hence we add $\Pr[\bar\Psi]$ to $\varepsilon'$) where $\varepsilon'=\frac 1 2 \cdot 2^{-\frac 1 2 ((\delta/2-2\kappa) n -1 -q-\ell)} + 2\cdot 2^{-\kappa n}$. 
Recall that $Z=F(X) \xor G(W)$. By security of the one-time pad it follows that 
\begin{equation}
\distuni(Z|E'FG,W=w,W'=w') \leq \varepsilon',
\label{eq:zuni}
\end{equation}
for any $w,w'$ such that $w\neq w'$.
%where  
%It follows that $Z:=F(X) \xor G(W)$ is $\varepsilon'$-close to uniformly distributed and independent of $F$, $G$, $W$, $W'$, and $E'$, conditioned on $W \neq W'$. 
To prove the claim, we need to bound,
\begin{align}
\nonumber\delta& ( \rho_{WW'E |W\neq W'} ,  \rho_{W \leftrightarrow W' \leftrightarrow E|W\neq W'} ) \\
%&=\tfrac12 \|\rho_{WW'E|W\neq W'}  - \rho_{W \leftrightarrow W' \leftrightarrow E|W\neq W'} \|_1 \\
\nonumber&=\tfrac12 \|\rho_{WW'E'FGZ|W\neq W'}  - \rho_{W \leftrightarrow W' \leftrightarrow E'FGZ|W\neq W'} \|_1 \\
\nonumber& \leq  \tfrac12 \|\rho_{WW'E'FGZ|W\neq W'}  -  \rho_{WW'E'FG|W\neq W'} \kron 2^{-\ell} \id \|_1 \\
\label{eq:triang}&\hspace{2em} + \tfrac12 \| \rho_{WW'E'FG|W\neq W'}\kron 2^{-\ell} \id  - \rho_{ W \leftrightarrow W' \leftrightarrow E'FGZ|W\neq W'} \|_1
\end{align}
where the equality follows by definition of trace distance (\refdef{tracedist}) and the fact that the output state $E$ is obtained by applying a unitary transformation to the set of registers ($E'$, $F$, $G$, $W'$, $Z$). The inequality is the triangle inequality; in the remainder of the proof, we will show that both terms in \refeq{triang} are upper bounded by $\varepsilon'$.
\begin{align*}
\tfrac12 & \|\rho_{WW'E'FGZ|W\neq W'}  -  \rho_{WW'E'FG|W\neq W'} \kron 2^{-\ell} \id \|_1\\
& = \sum_{w \neq w'} P_{WW'|W\neq W'}(w,w')\, \distuni(Z|E'FG,W=w,W'=w')\leq \varepsilon',
\end{align*}
where the latter inequality follows from \refeq{zuni}.%$iby privacy amplification. 
For the other term, we reason as follows: 
\begin{align*}
\tfrac12 & \| \rho_{WW'E'FG|W\neq W'}\kron 2^{-\ell} \id  - \rho_{ W \leftrightarrow W' \leftrightarrow E'FGZ|W\neq W'} \|_1\\
& =
\tfrac12 \sum_{w \neq w'} P_{WW'|W\neq W'}(w,w') \, \|\rho^{w,w'}_{E'FG|W\neq W'} \kron 2^{-\ell} \id - \rho^{w'}_{E'FGZ|W\neq W'} \|_1 \\
& =\tfrac12 \sum_{w \neq w'} P_{WW'|W\neq W'}(w,w') \, \|\rho^{w,w'}_{E'FG|W\neq W'} \kron 2^{-\ell} \id \\
&\hspace{2em} - \hspace{-1em}\sum_{\substack{w''\\\text{ s.t. }w''\neq w'}}\hspace{-1em}P_{W|W',W\neq W'}(w''|w') \rho^{w'',w'}_{E'FGZ|W\neq W'} \|_1 \\
& =\tfrac12 \sum_{w'} P_{W'|W\neq W'}(w') \, \| \sum_{\substack{w\\\text{ s.t. }w\neq w'}} P_{W|W',W\neq W'}(w|w') \rho^{w,w'}_{E'FG|W\neq W'} \kron 2^{-\ell} \id \\
&\hspace{2em} - \hspace{-1em}\sum_{\substack{w''\\\text{ s.t. }w''\neq w'}}\hspace{-1em}P_{W|W',W\neq W'}(w''|w') \rho^{w'',w'}_{E'FGZ|W\neq W'}  \hspace{-1em}\sum_{\substack{w\\\text{ s.t. }w\neq w'}}\hspace{-1em}P_{W|W',W\neq W'}(w|w')         \|_1 \\
& = \tfrac12 \sum_{w \neq w'} P_{WW'|W\neq W'}(w,w') \, \|\rho^{w,w'}_{E'FG|W\neq W'} \kron 2^{-\ell} \id - \rho^{w,w'}_{E'FGZ|W\neq W'} \|_1 \\
& =\sum_{w \neq w'} P_{WW'|W\neq W'}(w,w') \, \distuni(Z|E'FG,W=w,W'=w') \leq  \varepsilon',
\end{align*}
where the first equality follows by definition of conditional independence and by a basic property of the trace distance; the third and fourth equality follow by linearity of the trace distance. The inequality on the last line follows from \refeq{zuni}. %by privacy amplification. 
This proves the claim.
\end{proof}

\chapter{User Security in the Single-Qubit-Operations Model} 
\label{sec:usecsqom}
We now consider a dishonest server \dishserver that can store an unbounded number of qubits. Clearly, against such a \dishserver, Theorem~\ref{thm:bqsm} provides no security guarantee anymore. We show here that there is still \emph{some} level of security left. Specifically, we show that \QID is still secure against a dishonest server \dishserver that can reliably store all the communicated qubits and measure them qubit-wise and non-adaptively at the end of the protocol. %, if he is restricted to store and measure them \emph{qubit-wise}.
This feature distinguishes our identification protocol from the protocol from~\cite{DFSS07}, which completely breaks down against such an attack. 

\section{The Model}
\label{sec:measmodel}
Formally, a dishonest server \dishserver in the SQOM is modeled as follows. 
\begin{enumerate}
\item \dishserver may reliably store the $n$-qubit state $\ket{x}_{\encoding(w)} = \ket{x_1}_{\encoding(w)_1} \otimes \cdots \otimes \ket{x_n}_{\encoding(w)_n}$ received in step (1) of \QID.
\item At the end of the protocol, in step (5), \dishserver chooses an arbitrary sequence $\theta = (\theta_1,\ldots,\theta_n)$, where each $\theta_i$ describes an arbitrary orthonormal basis of $\mathbb{C}^2$, and measures each qubit $\ket{x_i}_{\encoding(w)_i}$ in basis $\theta_i$ to observe $Y_i \in \set{0,1}$. Hence, we assume that \emph{\dishserver measures all qubits at the end of the protocol.}
\item The choice of $\theta$ may depend on all the classical information gathered during the execution of the protocol, but we assume a \emph{non-adaptive} setting where $\theta_i$ does not depend on $Y_j$ for $i \neq j$, i.e., \dishserver has to choose $\theta$ entirely before performing any measurement. \end{enumerate}

Considering complete projective measurements acting on individual qubits, rather than general single-qubit POVMs, may be considered a restriction of our model. Nonetheless, general POVM measurements can always be described by projective measurements on a bigger system. In this sense, restricting to projective measurements is consistent with the requirement of single-qubit operations. It seems non-trivial to extend our security proof to general single-qubit POVMs. 

The restriction to non-adaptive measurements (item 3) is rather strong, even though the protocol from \cite{DFSS07} already breaks down in this non-adaptive setting. The restriction was introduced as a stepping stone towards proving the adaptive case. Up to now, we have unfortunately not yet succeeded in doing so, hence we leave the adaptive case for future research. 

We also leave for future research the case of a less restricted dishonest server \dishserver that can do measurements on blocks that are less stringently bounded in size. 
Whereas the adaptive versus non-adaptive issue appears to be a proof-technical problem (\QID looks secure also against an adaptive \dishserver), allowing measurements on larger blocks will require a new protocol, since \QID becomes insecure when \dishserver can do measurements on blocks of size $2$, as we show in \refsec{attack}.

\section{No Privacy Amplification}
One might expect that proving security of \QID in the SQOM, i.e., against a dishonest server \dishserver that is restricted to single-qubit operations should be straightforward, but actually the opposite is true, for the following reason. Even though it is not hard to show that after his measurements, \dishserver has lower bounded uncertainty in $x$ (except if he was able to guess $w$), it is not clear how to conclude that $f(x)$ is close to random so that $z$ does not reveal a significant amount of information about $w$. The reason is that standard privacy amplification fails to apply here. Indeed, the model allows \dishserver to postpone the measurement of all qubits to step (5) of the protocol. The hash function $f$, however, is chosen and sent already in step (3). This means that \dishserver can choose his  measurements in step (5) depending on $f$. As a consequence, the distribution of $x$ from the point of view of \dishserver may depend on the choice of the hash function $f$, in which case the privacy-amplification theorem does not give any guarantees.

\section{Single-Qubit Measurements}
%\todo{introduce \emph{single-qubit measurement} terminology? is it needed?}
%In this paragraph, we recall a few well-known facts about single-qubit measurements, and introduce some convenient terminology and notation. 
%A \todo{weird name: single-qubit measurement for an $n$-qubit system} is given by a sequence
Consider an arbitrary sequence $\theta = (\theta_1,\ldots,\theta_n)$ where each $\theta_i$ describes an orthonormal basis of $\cnum^2$.
%orthonormal two-outcome measurements, each one acting on a single qubit. 
Let $\ket{\psi}$ be an $n$-qubit system of the form 
\[
\ket{\psi} = \ket{x}_b=H^{b_1}\ket{x_1} \kron \cdots \kron H^{b_n} \ket{x_n},
\]
where $x$ and $b$ are arbitrary in $\set{0,1}^n$. Measuring $\ket{\psi}$ qubit-wise in basis $\theta$ results in a measurement outcome $Y = (Y_1,\ldots,Y_n) \in \set{0,1}^n$. 
%Think of Suppose that 
Suppose that 
$x$, $b$ and $\theta$ are in fact realizations of the random variables $X$, $B$ and $\Theta$ respectively. It follows immediately from the product structure of the state $\ket{\psi}$ that
\[
P_{Y|XB\Theta}(y|x,b,\theta) = \prod_{i=0}^n P_{Y_i|X_i B_i \Theta_i}(y_i|x_i,b_i,\theta_i),
\]
i.e. the random variables $Y_i$ are statistically independent conditioned on arbitrary fixed values for $X_i$, $B_i$ and $\Theta_i$ but such that $P_{X_iB_i\Theta_i}(x_i,b_i,\theta_i)>0$.

\begin{lemma}
\label{lem:probsymmetry}
The distribution $P_{Y_i|X_i B_i \Theta_i}(y_i | x_i , b_i, \theta_i)$ exhibits the following symmetries:
\[
P_{Y_i|X_i B_i \Theta_i}(0 | 0 , b_i, \theta_i) = P_{Y_i|X_i B_i \Theta_i}(1 | 1 , b_i, \theta_i)
\]
and 
\[
P_{Y_i|X_i B_i \Theta_i}(0 | 1 , b_i, \theta_i) = P_{Y_i|X_i B_i \Theta_i}(1 | 0 , b_i, \theta_i)
\]
for all $i \in [n]$, for all $b_i$ and $\theta_i$ with $P_{X_iB_i\Theta_i}(\xi,b_i,\theta_i)>0$ for all $\xi \in \set{0,1}$.
\end{lemma}
The proof can found in \refapp{prfprobsym}.
\begin{comment}

\begin{align}
\nonumber P&_{Y_i|X_i B_i \Theta_i}(0  | x_i , b_i, \theta_i)  = | \bra{x_i}H^{b_i} (\alpha \ket{0}+\beta \ket{1})|^2 \\%(\alpha \bra{0}+\beta \bra{1})H^{b_i}\ket{x_i}
\label{eq:yiszero} &   =\Big|\alpha + b_i\Big(\frac{\alpha + \beta}{\sqrt2}-\alpha\Big)\Big|^2 \braket{x_i}{0} +\Big|  \beta + b_i \Big(\frac{\alpha - \beta}{\sqrt2}-\beta \Big)\Big|^2  \braket{x_i}{1} %(\alpha \bra{0}+\beta \bra{1})H^{b_i}\ket{x_i}
\end{align}
and
\begin{align}
\nonumber P&_{Y_i|X_i B_i \Theta_i}(1  | x_i , b_i, \theta_i)  = | \bra{x_i}H^{b_i} (\beta \ket{0}-\alpha \ket{1})|^2 \\%(\alpha \bra{0}+\beta \bra{1})H^{b_i}\ket{x_i}
\label{eq:yisone} &   =\Big|\beta + b_i\Big(\frac{\beta-\alpha}{\sqrt2}-\beta\Big)\Big|^2 \braket{x_i}{0} +\Big|  - \alpha + b_i \Big(\frac{\alpha + \beta}{\sqrt2}+\alpha \Big)\Big|^2  \braket{x_i}{1}. %(\alpha \bra{0}+\beta \bra{1})H^{b_i}\ket{x_i}
\end{align}
By substituting any choice for $b_i \in \set{0,1}$, we see that the coefficient in front of $\braket{x_i}{0}$ in \refeq{yiszero} equals the coefficient in front of $\braket{x_i}{1}$ in \refeq{yisone}, and \emph{vice versa}. Furthermore, we have that $\braket{x_i}{0} = 1$ if $x_i=0$ and $\braket{x_i}{0}$ vanishes when $x_i=1$, and the opposite holds for $\braket{x_i}{1}$. Hence the claim follows.
\end{comment}
The symmetry characterized in \reflem{probsymmetry} coincides with that of the \emph{binary symmetric channel}, 
%is the same 
%at of a binary symmetric channel, 
i.e. we can view $Y$ as a ``noisy version'' of $X$, where this noise---produced by the measurement---is independent of $X$.

% produced by the measurement, and is independent of $x$. 
Formally, %for any realization $x$ of $X$, 
we can write $Y$ as 
\begin{equation}
\label{eq:mmodel}
 Y = X \oplus \Delta,
\end{equation}
where the random variable $\Delta=(\Delta_1,\ldots,\Delta_n) \in \set{0,1}^n$ thus represents the error between the random variable $X\in \set{0,1}^n$ that is ``encoded'' in the quantum state and the measurement outcome $Y\in \set{0,1}^n$. %As a consequence of the qubit-wise product structure of the state $\ket{\psi}$ and of the measurement, it holds that the joint distribution 
%\todo{do we need indep 
%$P_{\Delta_i\Delta_j} = P_{\Delta_i}P_{\Delta_j}$ for all $i,j \in [n]$, i.e.\ the $\Delta_i$ are statistically independent of each other.
% the $\delta_i$'s are independent of each other. 
%Furthermore, 
By substituting \refeq{mmodel} in \reflem{probsymmetry}, we get the following corollary.
\begin{corollary}[Independence Between $\Delta$ and $X$]
%The random variable $\Delta$ is statistically independent of $X$, i.e.
For every $i\in [n]$ it holds that
\[
P_{\Delta_i|X_i B_i \Theta_i}(\delta_i  | x_i , b_i, \theta_i)  =  P_{\Delta_i| B_i \Theta_i}(\delta_i  | b_i, \theta_i)
\]
for all $\delta_i \in \set{0,1}$ and for all $x_i$, $b_i$ and $\theta_i$ such that $P_{X_iB_i\Theta_i}(x_i,b_i,\theta_i)>0$.
\label{cor:indepxd}
\end{corollary}
Furthermore, since the random variables $Y_i$ are statistically independent conditioned on fixed values for $X_i$, $B_i$ and $\Theta_i$, it follows that the $\Delta_i$ are statistically independent conditioned on fixed values for $B_i$ and $\Theta_i$.

\begin{definition}[Quantized Basis]
\label{def:quantb}
For any orthonormal basis $\theta_i=\set{\ket{v_1},\ket{v_2}}$ on $\cnum^2$, we define the \emph{quantized basis} of $\theta_i$ as 
\[
\hat{\theta}_i := j^* \in \set{0,1},\quad \text{where } j^* \in \operatorname*{arg\,max}_{j\in \set{0,1}} \max_{k\in\set{1,2}} |\bra{v_k}H^j \ket{0}|.
\]
If both $j\in \set{0,1}$ attain the maximum, then $j^*$ is chosen arbitrarily. The quantized basis of the sequence $\theta = (\theta_1, \ldots, \theta_n)$ is naturally defined as the element-wise application of the above, resulting in $\hat{\theta} \in \bat^n$.
\end{definition}

We will use the bias as a measure for the predictability of $\Delta_i$. %Recall that When 

%We will make use of the following relation. 
\begin{thm}
When measuring the qubit $H^{b_i} \ket{x_i}$ for any $x_i,b_i \in \set{0,1}$ in any orthonormal basis $\theta_i$ on $\cnum^2$ for which the quantized basis $\hat{\theta}_i$ is the complement of $b_i$, i.e.\ $\hat{\theta}_i = b_i \xor 1$, then the bias of $\Delta_i \in \set{0,1}$, where $\Delta_i = Y_i \xor x_i$ and  $Y_i \in \set{0,1}$ is the measurement outcome, is upper bounded by 
\[
\bias(\Delta_i) \leq \frac{1}{\sqrt2}.
\]
\label{thm:biasbound}
\end{thm}
Since the theorem holds for any $x_i \in \set{0,1}$ and since \refcor{indepxd} guarantees that $\Delta_i$ is independent from an arbitrary random variable $X_i$, the theorem also applies when we replace $x_i$ by the random variable $X_i$. 

In order to prove \refthm{biasbound}, we need the following lemma. 
\begin{lemma}
\label{lem:biasrelation}
If, for any orthonormal basis $\theta_i$ on $\cnum^2$, there exists a bit $b_i \in \set{0,1}$ so that when measuring the qubit $H^{b_i} \ket{x_i}$ for any $x_i \in \set{0,1}$ in the basis $\theta_i$ to obtain $Z_i\in \set{0,1}$
it holds that
\[
\bias(Z_i) \geq 1/\sqrt2, 
\]
then it holds that when measuring the qubit $H^{b_i \xor 1} \ket{x_i}$ in the basis $\theta_i$ to obtain $Y_i~\in~\set{0,1}$,
\[
\bias(Y_i) \leq 1/\sqrt2.
\]
\end{lemma}
\begin{proof}
First note that for any $x_i,b_i \in \set{0,1}$ and any orthonormal basis $\theta_i$ on $\cnum^2$, measuring a state $H^{b_i}\ket{x_i}$ in $\theta_i=\set{\ket{v},\ket{w}}$ where $\ket{v} = \alpha \ket{0}+\beta\ket{1}$ and $\ket{w}=\beta\ket{0}-\alpha\ket{1}$ gives the same outcome distribution (up to permutations) as when measuring one of the basis states of $\theta_i$ (when viewed as a quantum state), say \ket{w}, using the basis $\set{H^{b_i}\ket{x_i},H^{b_i} \ket{x_i \oplus 1}}$. 
To see why this holds, note that it follows immediately that $|\bra{w}H^{b_i}\ket{x_i}|^2 = |\bra{x_i}H^{b_i}\ket{w}|^2$. Furthermore, we have already shown in the proof of \reflem{probsymmetry} that 
%\begin{equation}
\[
|\bra{v}H^{b_i}\ket{x_i}|^2 = |\bra{w}H^{b_i}\ket{x_i\oplus 1}|^2
\]
%\label{eq:vw}
%\end{equation}
holds. %but note that we have already shown this in the proof of \reflem{probsymmetry}.

Hence, we can apply \refthm{morehadamard} with $\rho = \proj{w}$ (this implies that $n=1$), $m~=~2$ and $\mcal{B}_0$ and $\mcal{B}_1$ are the computational and Hadamard basis respectively. The maximum overlap between those bases is $c=1/\sqrt2$. \refthm{morehadamard} gives us that
\[
p^{\set{\ket{0},\ket{1}}}_{\max} + p^{\set{\ket{+},\ket{-}}}_{\max}  \leq 1 + \frac{1}{\sqrt2},
\]
where $p^{\set{\ket{0},\ket{1}}}_{\max}$ and  $p^{\set{\ket{+},\ket{-}}}_{\max}$ respectively denote the maximum probability in the distribution obtained by measuring in the computational and Hadamard basis. By simple manipulations we can write this as a bound on the sum of the biases:
\begin{align}
%2( p^{\set{\ket{0},\ket{1}}}_{\max} &+ p^{\set{\ket{+},\ket{-}}}_{\max})  \leq 2 + \frac{2}{\sqrt2}\\
\nonumber  \frac{2}{\sqrt2} & \geq (2 p^{\set{\ket{0},\ket{1}}}_{\max} - 1) + (2 p^{\set{\ket{+},\ket{-}}}_{\max} -1)  \\
%\nonumber&= \bias(Y_i | B_i = 0 ) + \bias(Y_i | B_i = 1 ) \\% \leq \frac{2}{\sqrt2}.
\label{eq:biassum}&= \bias(Y_i ) + \bias(Z_i) .% \leq \frac{2}{\sqrt2}.
\end{align}
From this relation, the claim follows immediately.
\end{proof}
%Although we do not need the following fact anywhere in this text, w
Following \cite{Schaffner07}, we want to remark that both biases in \refeq{biassum} are equal to $1/\sqrt2$ when $\theta_i$ is the \emph{Breidbart basis}, which is the basis that is precisely ``in between'' the computational and the Hadamard basis:\footnote{In \cite{Schaffner07}, the corresponding state is called the ``Hadamard-invariant state.''}
\[
\ket{v} = %\frac{\ket{0}+\ket{+}}{\big|\ket{0}+\ket{+}\big|} = %\frac {\sqrt{2+\sqrt2}}{2}
\cos(\tfrac{\pi}{8})
\ket{0} +
%\frac{1}{\sqrt{2 \big(2+\sqrt{2}\big)}} 
\sin (\tfrac{\pi}{8}) \ket{1}
\qquad\text{and}\qquad
%and
%\[
% \qquad \text{and}\qquad
\ket{w}= %\frac{\ket{-}-\ket{1}}{\big|\ket{-}-\ket{1}\big|} = %\frac{1}{\sqrt{2 \big(2+\sqrt{2}\big)}}
\sin (\tfrac{\pi}{8})
 \ket{0} -  \cos (\tfrac{\pi}{8})
%\frac {\sqrt{2+\sqrt2}}{2}
\ket{1}.
\]

\begin{proof}[Proof of \refthm{biasbound}]
Let $\theta_i=\set{\ket{v_0},\ket{v_1}}$. We will make a case distinction based on the value of
\begin{equation}
\mu:= \max_{k\in \set{0,1}} |\bra{v_k}H^{\hat\theta_i} \ket{0}|.
\label{eq:maximization}
\end{equation}
If $\mu \leq \cos(\pi/8)$, then we also have that $\max_{k\in \set{0,1}}|\bra{v_k}H^{b_i} \ket{x_i}| \leq \cos(\pi/8)$ where $b_i=\hat \theta_i \xor 1$, this holds by definition of the quantized basis (\refdef{quantb}). %(If this would not hold then $b$ itself would be the quantized basis of $\theta$.)
Then, the probability of obtaining outcome $Y_i=k^*$, where $k^* \in \set{0,1}$ achieves the maximum in \refeq{maximization}, is bounded by
\[
P_{Y_i}(k^*) =  |\bra{v_{k^*}}H^{b_i} \ket{x_i}|^2 \leq \cos^2(\pi/8) = \tfrac12 + \tfrac{1}{2\sqrt2}. 
\]
Hence,
\[
\bias(\Delta_i) = \bias(Y_i) = | P_{Y_i}(k^*) - (1-P_{Y_i}(k^*))| = |2 P_{Y_i}(k^*) - 1| \leq \tfrac{1}{\sqrt2}.  
\]
If $\mu > \cos(\pi/8)$, then when measuring the state $H^{\hat\theta_i}\ket{x_i}$ in $\theta_i$ to obtain $Z_i~\in~\set{0,1}$, we have that $\bias(Z_i) > 1/\sqrt2$ (this follows from similar computations as performed above).
We now invoke \reflem{biasrelation} to conclude that when measuring the state $H^{b_i}\ket{x_i}$ in $\theta_i$ to obtain $Y_i$, 
$\bias(\Delta_i) = \bias(Y_i) < \tfrac{1}{\sqrt2}$.
\end{proof}

\section{User Security of \QID} 
\label{sec:usecqw}
We are now ready to state and prove the security of \QID against a dishonest user in the SQOM.

\begin{thm}[User Security]\label{thm:usec}
Let \dishserver be a dishonest server with unbounded quantum storage that is restricted to non-adaptive single-qubit operations, as specified in \refsec{measmodel}. Then, for any $0 < \beta < \tfrac14$, user security (as defined in \refdef{usec}) holds with
\[
\textstyle\varepsilon \leq \tfrac12 2^{\frac{1}{2} \ell-  \frac14 (\frac14 -\beta)d} + {m\choose 2} 2^{2\ell}\exp(-2d \beta^2)
\]
\end{thm}
Note that $d$ is typically linear in $n$ whereas $\ell$ is chosen independently of $n$, hence the expression above is negligible in $d$. 

To prove \refthm{usec} we need the following technical lemma and corollary. 
Recall that $\mathcal{F}$ denotes the class of all linear functions from $\set{0,1}^n$ to $\set{0,1}^\ell$, where $\ell < n$, represented as binary $\ell \times n$ matrices. 

\begin{lemma}
\label{lem:schur}
Let $n$, $k$ and $\ell$ be arbitrary positive integers, let $0<\beta<\tfrac14$ and let $\mathcal I \subset \setn$ such that  $|\mathcal I|\geq k$, and let $F$ be uniform over $\mathcal{F}= \bat^{\ell \times n}$. Then, it holds except with probability $2^{2\ell}\exp(- 2 k \beta^2)$ (the probability is over the random matrix $F$) that
\[
\big|(f \schur g )_\Tset{I}\big|
> (\tfrac14-\beta) k \qquad \forall f,g \in \lswon
%\setminus \set{0^n}% \text{ s.t. } f \neq g
\]
\end{lemma}

\begin{proof}
Without loss of generality, we will assume that $|\Tset{I}| =k$. Now take arbitrary but non-zero vectors $r,s \in \set{0,1}^{\ell}$ and let $V:=rF$ and $W:=sF$.
We will analyze the case $r \neq s$; the case $r=s$ is similar but simpler. Because each element of $F$ is an independent random bit, and $r$ and $s$ are non-zero and $r\neq s$, $V$ and $W$ are independent and uniformly distributed $n$-bit vectors with expected relative Hamming weight $1/2$.
Hence, on average $|(V \schur W)_\Tset{I}|$ equals $k/4$.
Furthermore, using Hoeffding's inequality (\refthm{hoeffding}), we may conclude that
\[
\Pr \bigg[  \frac{k}{4}  - |(V \schur W)_\Tset{I}| > \beta k \bigg] = \Pr \bigg[|(V \schur W)_\Tset{I}| < \big(\tfrac14 - \beta\big) k  \bigg] \leq \exp(- 2 k \beta^2) \, .
\]
Finally, the claim follows by applying the union bound over the choice of $r$ and $s$ (each $2^\ell$ possibilities).
\end{proof}

Recall that \Tset{C} is a binary code with minimum distance $d$, $\encoding(\cdot)$ its encoding function, and that $m:=|\Tset{W}|$.

\begin{corollary}
\label{cor:atmost}
Let $0 <\beta <\tfrac14$, and let $F$ be uniformly distributed over $\Tset{F}$. Then, $F$ has the following property except with probability $ \binom{m}{2} 2^{2\ell}\exp(-2d\beta^2)$: for any string $s \in \set{0,1}^n$ (possibly depending on the choice of~$F$), there exists at most one $\tilde c \in \Tset{C}$ such that for any code word $c \in \Tset{C}$ different from $ \tilde c$, it holds that 
\[
| f \schur (c \xor s) | \geq \tfrac12(\tfrac14-\beta)d
 \qquad  \forall f \in \lswon 
\]
\end{corollary}
We prove the statement by arguing for two $\tilde c$'s and showing that they must be identical.
In the proof, we will make use of the two following propositions.

%elementary properties of the Schur product and the Hamming weight:
\begin{prop}
$|a| \geq |a \schur b| $ for all $a,b \in \set{0,1}^n$.
\label{prop:pone}
\end{prop}
\begin{proof} Follows immediately.\end{proof}
\begin{prop}
$|a\schur b | + |a \schur c| \geq |a \schur (b \xor c)|$ for all $a,b,c \in \set{0,1}^n$.
\label{prop:ptwo}
\end{prop}
\begin{proof}
$ |a \schur (b \xor c)|  = |a \schur b \xor a\schur c| \leq  |a \schur b | + | a\schur c|$, where the equality is the distributivity of the Schur product, and the inequality is the triangle inequality for the Hamming weight.
\end{proof}

\begin{proof}[Proof of Corollary~\ref{cor:atmost}]
%Let $\Tset{I}:= \set{  i \in \setn: c_i \neq c'_i}$ for any $c,c' \in \C$. 
By \reflem{schur} with $\Tset{I}:= \set{  i \in \setn: c_i \neq c'_i}$ for $c,c' \in \C$, and by applying the union bound over all possible pairs $(c,c')$, we obtain that except with probability $\binom{m}{2} 2^{2\ell} \exp(-2d\beta^2)$ (over the choice of~$F$), it holds that 
\begin{equation}
|f \schur g \schur (c\oplus c')|>(\tfrac14-\beta)d
\label{eq:contrad}
\end{equation}
for all $f,g \in \lswon$ and all $c,c' \in \C$ with $c\neq c'$.

Now, for such an $F$, and for every choice of $s\in \bat^n$, consider $\tilde c_1, \tilde c_2 \in \C$ and $f_1, f_2 \in \lswon$ such that
\[ |f_1 \schur (\tilde c_1 \xor s)| <\tfrac12(  \tfrac14 -\beta)d \quad
%, and take $\tilde c_2 \in \C$ and $f_2 \in \linspan(F)$ such that
\text{and}\quad |f_2 \schur (\tilde c_2 \xor s)| < \tfrac12(  \tfrac14 -\beta)d. 
\]
We will show that this implies $\tilde c_1 = \tilde c_2$, which proves the claim. 
%and let $\tilde{ \mathcal{I}}:= \set{  i \in \setn: \tilde c_{1,i} \neq \tilde c_{2,i}}$. 
Indeed, we can write
\begin{align*}
(\tfrac14- \beta)d & > |f_1 \schur (\tilde c_1 \xor s)| + |f_2 \schur (\tilde c_2 \xor s)| \\
&\geq |f_1 \schur f_2 \schur (\tilde c_1 \xor s)| + |f_1 \schur f_2 \schur (\tilde c_2 \xor s)| 
 \geq |f_1 \schur f_2 \schur (\tilde c_1\! \xor \!\tilde c_2)| 
\end{align*}
where the second inequality is \refprop{pone} applied twice  and the third inequality is \refprop{ptwo}. This contradicts \refeq{contrad} unless $\tilde c_1 = \tilde c_2$.
% From this we conclude that $\tilde c_1 = \tilde c_2$, otherwise we get a contradiction.
% Hence, except with probability $\leq \binom{m}{2} 2^{2\ell}\exp(-2d\beta^2)$, for any $s$ there can only be at most one $\tilde c$ for which there exists an $f \in \lswon$ such that $| f \schur (\tilde c \xor s) | < \tfrac12 (\tfrac14 - \beta )d$.
\end{proof}
Now we are ready to prove \refthm{usec}. In the proof, when $F \in \mcal{F}$ acts on an $n$-bit vector $x \in \bat^n$, we prefer the notation $F(x)$ over matrix-product notation $Fx$.\footnote{When using matrix-product notation ambiguities could arise, e.g.\ in subscripts of probability distributions like $P_{FX}$: then it is not clear whether this means the joint distribution of $F$ and $X$ or the distribution of $F$ acting on $X$?} 

\begin{proof}[Proof of Theorem \ref{thm:usec}]
Consider an execution of \QID, with a dishonest server \dishserver as described in \refsec{measmodel}.
% above, who stores all the qubits received in step (2), and measures them qubit-wise at the end of the execution, depending on what he has seen during the execution of the protocol. 
We let $W, X$ and $Z$ be the random variables that describe the values $w, x$ and $z$ occurring in the protocol. 

From \QID's description, we see that $F$ is uniform over $\mcal{F}$. Hence, by \refcor{atmost} it will be ``good'' (in the sense that the bound from \refcor{atmost} holds) except with probability ${m\choose 2} 2^{2\ell}\exp(-2d \beta^2)$.
From here, we consider a fixed choice for $F$ and condition on the event that it is ``good,'' %and treat it as being ``good,'' 
we thus book-keep the probability that $F$ is ``bad'' and take it into account at the end of the analysis. 
%We now fix $F$ and 
%
%, namely one which is  ``good'' in the sense that the bound from \refcor{atmost} holds. By \refcor{atmost}, this is the case except with probability ${m\choose 2} 2^{2\ell}\exp(-2d \beta^2)$, an error probability that we book-keep and take into account at the end of the analysis. 
Although we have fixed $F$, we will keep using capital notation for it, to emphasize that $F$ is a matrix.
We also fix $G=g$ for an arbitrary $g$; the analysis below holds for any such choice. 

Let $\Theta$ describe the qubit-wise measurement performed by \dishserver at the end of the execution, and $Y$ the corresponding measurement outcome. By the non-adaptivity restriction and by the requirement in \refdef{usec} that \dishserver is initially independent of $W$, we may conclude that, once $G$ and $F$ are fixed, $\Theta$ is a function of $Z$. (Recall that $Z = F(X)\xor g(W)$.) 
%Hence, \dishserver's information at the end of the protocol consists of $E = (Y,Z,\Theta)$. 
%We recall that $\Theta$ is not restricted to be the qubit-wise measurement in the computational or the Hadamard basis, but it may be a measurement in an arbitrary basis for each qubit. %Nevertheless, it is guaranteed that conditioned on any value for $X$, the $Y_i$'s are independently distributed. 

We will define $W'$ with the help of \refcor{atmost}. Let $\q{\Theta}$ be the quantized basis of $\Theta$, as defined in \refdef{quantb}. 
%, by specifying how .
 Given a fixed value $\theta$ for $\Theta$, and thus a fixed value $\q{\theta}$ for $\q{\Theta}$, we set $s$, which is a variable that occurs in \refcor{atmost}, to $s = \q{\theta}$. \refcor{atmost} now guarantees that there exists \emph{at most one} $\tilde c$. If $\tilde c$ indeed exists, then we choose $w'$ such that $\encoding(w') = \tilde c$. Otherwise, we pick $w'\in \mcal{W}$ arbitrarily (any choice will do). Note that this defines the random variable $W'$, and furthermore note that $Z\rightarrow \Theta \rightarrow \hat \Theta\rightarrow W'$ forms a Markov chain.
%is a function of $\hat{\Theta}$ (and $F$). 
Moreover, by the choice of $w'$ it immediately follows from \refcor{atmost} that for all $w \neq w'$ and for all $f \in \lswon$ it holds that 
\begin{equation}
\big|f \odot (\encoding(w) \xor \q{\theta})\big|\geq \tfrac12(\tfrac14 -\beta)d.
\label{eq:hbound}
\end{equation}
We will make use of this bound later in the proof.

Since the model (\refsec{measmodel}) enforces the dishonest server to measure all qubits at the end of the protocol,
%\footnote{Note that for QID this assumption makes sense; 
%after termination of the protocol there is no  the is a sensible assumption, 
%
%since all the information  the quantum state is only 
%relevant  
%; 
%in QKD, on the other hand, where the key that is derived from the quantum state is going to be used after termination of the protocol, this assumption cannot be made.} 
the system $E=(Y,Z,\Theta)$ is classical and hence the trace-distance-based user-security definition (\refdef{usec}) simplifies to a bound on the statistical distance between distributions. %To prove the security claim, 
I.e., it is sufficient to prove that
\[
\stdist(P_{E W |W'=w',W'\neq W},  P_{W|W'=w',W\neq W'} P_{E|W'=w',W\neq W'})\leq \varepsilon % \quad \forall w'
\]
holds for any $w'$. %which we do below. 
Consider the distribution that appears above as the first argument to the statistical distance, i.e. $P_{E W |W'=w',W'\neq W}$. By substituting $E=(Y,Z,\Theta)$, it factors as follows\footnote{Note that we shorten notation here by omitting the parentheses containing the function arguments. The quantification is over all inputs for which all involved conditional probabilities are well-defined.} % \refrem{sloppynot} applies here.} %We use somewhat sloppy notation; by %the equality 
%$P_{Y | Z\Theta WW', W \neq W'} = P_{Y | F(X)\Theta WW', W \neq W'}$ we mean
%should be understood as 
%$P_{Y | Z\Theta WW', W \neq W'}(y|z,\theta,w,w') = P_{Y | F(X)\Theta WW', W \neq W'}(y|z\xor g(w),\theta,w,w')$ for all $y,z,\theta,w,$ with $w \neq w'$ where $w'$ is determined by $\theta$. }
\begin{align}
\nonumber P_{YZ\Theta W |W', W \neq W'} 
&= P_{W| W', W \neq W'}\ P_{Z \Theta |WW', W \neq W'}\ P_{Y | Z\Theta WW', W \neq W'} \\
&= P_{W| W', W \neq W'}\ P_{Z \Theta |W', W \neq W'}\ P_{Y | F(X)\Theta WW', W \neq W'}, 
\label{eq:factor}
\end{align}
where the equality $P_{Z\Theta |WW', W \neq W'} = P_{Z\Theta |W', W \neq W'}$ holds by the following argument: 
%the last equality follows from the fact that 
$Z$ is independent of $W$ (since $F(X)$ acts as one-time pad) and $Z \rightarrow \Theta \rightarrow W'$ is a Markov chain, and \dishserver (who computes $\Theta$ from $Z$) is initially independent of $W$ by \refdef{usec}, hence $W$ is independent of $Z$, $\Theta$ and $W'$, which implies the above equality. The equality $P_{Y | Z \Theta WW', W \neq W'} = P_{Y | F(X)\Theta WW', W \neq W'}$ holds by the observation that given $W$, $Z$ is uniquely determined by $F(X)$ and vice versa.%

In the remainder of this proof we will show that 
%der of this proof we show that
\[
%P_{Y |F(X)WW'}(y|u,w,w') = 
\distuni (Y | F(X)=u,\Theta=v,W=w,W'=w') 
\leq \tfrac12 2^{\frac{\ell}{2} - \frac14 (\frac14 -\beta)d},
%\leq \tfrac12 2^{\frac{\ell}{2} - \frac{(1-4\beta)d )}{16}} ,   %}(y|u,w,w')
\]
for all $u,v,w$ such that $w \neq w'$, where $w'$ is determined by $v$. 
%is exponentially close to the uniform distribution for any 
This then implies that the rightmost factor in \refeq{factor} is essentially independent of $W$, and concludes the proof.

To simplify notation, we define $\mathcal{E}$ to be the event
\[
\mcal{E}:=  \set{ F(X) = u, \Theta=v, W = w, W' = w'}
\]
for fixed but arbitrary choices $u$, $v$ and $w$ such that $w \neq w'$, where $w'$ is determined by $v$. We show closeness to the uniform distribution by using the XOR inequality from Diaconis \etal (\refthm{diaconis}), i.e., we use the inequality  
$$
%\SD(P_{Y|\Tset{E}}, U_{\set{0,1}^n})
\distuni(Y|\mcal{E}) \leq \tfrac12 \Big[\sum_{\alpha} \bias(\alpha\cdot Y|\Tset{E})^2\Big]^\frac12,
$$
where the sum is over all $\alpha$ in $\set{0,1}^n \setminus \set{0^n}$. We split this sum into two parts, one for $\alpha \in \linspan(F)$ and one for $\alpha$ not in $\linspan(F)$, and analyze the two parts separately.

Since $X$ is uniformly distributed, it follows that for any $\alpha \notin \linspan(F)$, it holds that $P_{\alpha \cdot X| F(X)}(\cdot|u) = \tfrac12$ (for any $u$). 
We conclude that 
\begin{align*}
\tfrac12 & = P_{\alpha\cdot X| F(X)} = P_{\alpha\cdot X| F(X) W}
= P_{\alpha\cdot X| F(X) \Theta WW' }  \\
&= P_{\alpha\cdot Y| F(X)  \Theta WW'  } = P_{\alpha\cdot Y|\mcal{E}}
%P_{\alpha\cdot Y|\mcal{E}}&= P_{\alpha\cdot Y| F(X)  \Theta WW'  } = P_{\alpha\cdot X| F(X) \Theta WW' }  \\
%&= P_{\alpha\cdot X| F(X) W} = =  
\quad \forall \alpha \notin \linspan(F).
\end{align*}
The second equality follows since $W$ is independent of $X$.
The third equality holds by the fact that $\Theta$ is computed from $F(X) \oplus g(W)$ and $W'$ is determined by $\Theta$.
The fourth equality follows by the security of the one-time pad, i.e. recall that $Y = X \oplus \Delta$, where by \refcor{indepxd} it holds that $\Delta\in \bat^n$ is independent of $X$ when conditioned
%Note that \refcor{indepxd} holds when conditioned 
on fixed values for $B=\encoding(W)$ and $\Theta$. %Note that we may indeed apply \refcor{indepxd} here (and in the other part of the proof below) for all $i\in [n]$ because the event \mcal{E} fixes $\Theta$ by definition, and, via $W=w$, $B$ is (implicitly) fixed as $B=\encoding(w)$.  
Hence, it follows that $\bias (\alpha \cdot Y| \Tset{E} ) = 0$ for $\alpha \notin \linspan(F)$. 
%The second equality follows from the fact that (the distribution of) $E$ is determined solely by $W$ and $\Theta$. 

For any non-zero $\alpha  \in \linspan(F)$, we can write
\begin{align*}
\bias ( \alpha \cdot Y|\mcal{E}) &=
\bias ( \alpha \cdot (X \xor \Delta ) |\mcal{E}) \\
&= \bias ( \alpha \cdot X \xor \alpha \cdot \Delta  |\mcal{E}) &\text{(distributivity of dot product)}\\
&= \bias ( \alpha \cdot X |\mcal{E}) \bias( \alpha \cdot \Delta  |\mcal{E}) &\text{(\refcor{indepxd})}\\
& \leq \bias( \alpha \cdot \Delta  |\mcal{E}) 
&\text{($\bias(\alpha \cdot X) \leq 1$)}
\\
& =  \prod_{i \in [n]} \bias( \alpha_i \cdot \Delta_i  |\mcal{E}) &\text{($\Delta_i$ independent)}\\
& =  \prod_{i \in [n]:\alpha_i = 1} \bias( \Delta_i  |\mcal{E}) \\
&\leq \prod_{\substack{i \in [n]:\alpha_i = 1 \\ \hat{\theta}_i = \encoding(w)_i \xor 1}} 2^{-\frac12}&\text{(\refthm{biasbound}) }
\\&= 2^{-\frac12 |\alpha \schur (\encoding(w) \xor \hat\theta)|} \leq 2^{-\frac14(\frac14-\beta)d}&\text{(by \refeq{hbound})} 
\end{align*}

Combining the two parts, we get 
\begin{align*}
\distuni(Y|\mcal{E}) %SD(P_{Y|\Tset{E}}, U_{\set{0,1}^n}) 
& \leq \tfrac12 \Big[\sum_{\alpha} \bias(\alpha\cdot Y|\Tset{E})^2\Big]^\frac12 \\
&=\tfrac12 \Big[\sum_{\alpha \in \lswon} \bias(\alpha \cdot Y|\Tset{E})^2 + 0\,\Big]^\frac12 \leq \tfrac12 2^{\frac{\ell}{2} - \frac14 (\frac14 -\beta)d} \, .
\end{align*}
Incorporating the error probability of having a ``bad'' $F$ completes the proof. 
\end{proof}

\section{Attack against {\QID} with Operations on Pairs of Qubits} %$2$-coherent operations}
\label{sec:attack}
We present an attack with which the dishonest server \dishserver can discard two passwords in one execution of \QID using coherent operations on pairs of qubits.  % (sometimes called ``$2$-coherent quantum operations''). 

Before discussing this attack, we first explain 
a straightforward strategy by which \dishserver can discard one password per execution: \dishserver chooses a candidate password $\hat w$ and measures the state $H^{\encoding(W)}\ket{X}$ qubit-wise in the basis $\encoding(\hat w)$ to obtain $Y$. \dishserver then computes $F(Y)\xor g(\hat w)$ and compares this to $Z=F(X) \xor g(W)$, which he received from the user. 
If indeed $Z = F(Y)\xor g(\hat w)$, then it is very likely that $W = \hat w$, i.e.\ that \dishserver guessed the password correctly.

%Then, if $W= \hat w$, the measurement outcome $Y$ will be equal to $X$ and in this case we have that $F(X) = F(Y)$ and $g(W) = g(\hat w)$. 
%Note that \dishserver can compute $F(Y)\xor g(\hat w)$ by himself, and he receives $Z=F(X) \xor g(W)$ from the user. 
%In the other direction, 

Let us now explain the attack, which is obtained by modifying the above strategy. The attack is based on the following observation \cite{DFSS05}: if \dishserver can perform Bell measurements on qubit pairs $\ket{x_1}_{a} \ket{x_2}_{a}$, for $a \in \set{0,1}$, then he can learn the parity of $x_1 \xor x_2$ for both choices of $a$ simultaneously. This strategy can also be adapted to determine both parities of a pair in which the first qubit is encoded in a basis that is opposite to that of the second qubit, i.e.\ by appropriately applying a Hadamard gate prior to applying the Bell measurement.

Let the first bit of $Z$ be equal to $f\cdot X \xor g(W)_1$,\footnote{By $g(W)_1$ we mean the first bit of $g(W)$.} where $f \in \lswon$. % be a function for which 
%From the side information $Z$, \dishserver knows the value of $f\cdot X$ for some $f \in \lswon$. 
Let $\hat w_1$ and $\hat w_2$ be two candidate passwords. With the trick from above, \dishserver can measure the positions in the set 
\[
\mcal{P}:=\Set{i \in [n]}{f_i = 1,\encoding(\hat w_1)_i = 1 \xor   \encoding(\hat w_2)_i}
\]  %support of $f$ where $\encoding(w_1)$ and $\encoding(w_2)$ differ 
\emph{pairwise} (assuming $|\mcal{P}|$ to be even) using Bell measurements, while measuring the positions where $\encoding(\hat w_1)$ and $\encoding(\hat w_2)$ coincide using ordinary single-qubit measurements. %(see \reffig{attack} for an example). 
This allows him to compute both ``check bits'' 
corresponding to both passwords \emph{simultaneously}, i.e.\ those check bits coincide with $f \cdot Y_1 \xor g(\hat w_1)_1$ and $f \cdot Y_2 \xor g(\hat w_2)_1$, where $Y_1$ and $Y_2$ are the  outcomes that \dishserver would have obtained if he had measured all qubits qubit-wise in either $\encoding(\hat w_1)$ or $\encoding(\hat w_2)$, respectively.
%
% the ``check bit'' $f\cdot Y$
%
% for both candidate passwords simultaneously. 
If both these check bits are different from the bit $Z_1$, %f \cdot X $, 
then \dishserver can discard both $w_1$ and $w_2$.

We have seen that in the \emph{worst case}, the attack is capable of discarding two passwords in one execution, and hence clearly violates the security definition. On \emph{average}, however, the 
attack seems to discard just one password per execution, i.e.\ a candidate password cannot be discarded if its check bit is consistent with $Z_1$, which essentially happens with probability $1/2$. This raises the question whether the security definition is unnecessarily strong, because it seems that not being able to discard more than one password on average would be sufficient. Apart from this, it might be possible to improve the attack, e.g.\ by selecting the positions where to measure pairwise in a more clever way, 
as to obtain multiple check bits (corresponding to multiple $f$s in the span of $F$) per candidate password, thereby increasing the probability of discarding a wrong candidate password.

\chapter{Conclusion}
We view our work related to \QID as a first step in a promising line of research, aimed at achieving security in multiple models simultaneously.
The main open problem in the context of the SQOM is to reprove our results in a more general model in which the dishonest server \dishserver can choose his basis adaptively. 
Also, it would be interesting to see whether similar results can be obtained in a model where the adversary is restricted to performing quantum operations on blocks of several qubits.

\bibliography{boumanbibtex}
\begin{appendix}

%%%%%%%%%%%%%%%%%%%%%%%%%%%%%%%%%%%%%%%%%%%%%%%%%%%%%%%%%%%%%%%%%%%%%%%%%%%%%
%%%%%%%%%%%%%%%%%%%%%%%%%%%%%%%%%%%%%%%%%%%%%%%%%%%%%%%%%%%%%%%%%%%%%%%%%%%%%

%%%%%%%%%%%%%%%%%%%%%%%%%%%%%%%%%%%%%%%%%%%%%%%%%%%%%%%%%%%%%%%%%%%%%%%%%%%%%
\chapter{Proof of an Operator Norm Inequality (Proposition~\ref{prop:morebases})}\label{sec:proofinequality}
%%%%%%%%%%%%%%%%%%%%%%%%%%%%%%%%%%%%%%%%%%%%%%%%%%%%%%%%%%%%%%%%%%%%%%%%%%%%%
%The proof of the uncertainty relation is based on the operator inequality stated below in Proposition~\ref{prop:morebases}. In order to derive this inequality, we 
We first recall some basic properties of the operator norm $\|A \| \assign \sup \|A \ket{\psi}\|$, where the supremum is over all norm-$1$ vectors $\ket{\psi} \in \H$.
First of all, it is easy to see that
\[ \left\| \begin{pmatrix} A & 0 \\ 0 & B \end{pmatrix} \right\| =\max\left\{\|A\|,\|B\|\right\}.
\]
Also, from the fact that $\| A \| = \sup |\bra{\psi}A\ket{\varphi}|$,  
where the supremum is over all norm-$1$ $\ket{\psi},\ket{\varphi} \in \H$, it follows that $\|A^*\|=\|A\|$, where $A^*$ is the Hermitian transpose of $A$, and thus that for Hermitian matrices $A$ and~$B$:
$$
\| AB \| = \|(AB)^*\|=\|B^*A^*\|=\|BA\| \, .
$$
Furthermore, if $A$ is Hermitian then $\| A \| = \lambda_{\max}(A) \assign \max\{|\lambda_j| : \lambda_j \mbox{ an eigenvalue of } A \}$. 
Finally, the operator norm is \emph{unitarily invariant}, i.e., $\|A\|=\|UAV\|$ for all $A$ and for all unitary $U,V$. 

% From an equivalent definition of the norm $\|A\|$ as $\sup\limits_{\braket{y}{y}=\braket{x}{x}=1}
% \!\!\!|\bra{y}A\ket{x}|$, it is easy to see that $\|A^*\|=\|A\|$. 
% For
% two Hermitian matrices $A$ and $B$, we have that $\| AB \| = \|(AB)^*\|=\|B^*A^*\|=\|BA\|$.  The operator norm is \emph{unitarily
%   invariant}, i.e. for all unitary $U,V$, $\|A\|=\|UAV\|$ holds. 

\begin{lemma} \label{lem:ineq}
Any two $n \times n$ matrices $X$ and $Y$ for which the products $XY$ and $YX$
are Hermitian satisfy
\[ \|XY\| = \|YX\| \] 
\end{lemma}
\begin{proof}
For any two $n \times n$ matrices $X$ and $Y$, $XY$ and $YX$ have the
same eigenvalues, see e.g. \cite[Exercise I.3.7]{Bhatia97}. Therefore,
$\| XY \| = \lambda_{\max}(XY) = \lambda_{\max}(YX) = \|YX\|$.
\end{proof}

We are now ready to state and prove the norm inequality. We recall that an orthogonal projector $P$ satisfies $P^2 = P$ and $P^* = P$.  

\begin{prop} \label{prop:morebases}
For orthogonal projectors $A_1, A_2, \ldots, A_m$, it holds that 
\begin{equation*} \label{eq:multiproj} 
%\bigg\| \sum_{i=1}^m A_i \bigg\| 
\big\| A_1+\ldots+A_m \big\| \leq 1 + (m-1) \cdot \max_{1\leq j< k
  \leq m}
\big\|A_j A_k\big\|.
\end{equation*}
\end{prop}

The case $m = 2$ was proven in~\cite{DFSS05}, adapting a technique by Kittaneh \cite{Kittaneh97}. We extend the proof to an arbitrary $m$. 

\begin{proof}
Defining
\[
X \assign \begin{pmatrix} A_1 & A_2 & \cdots & A_m
    \\ 0 & 0 & \cdots & 0 \\ \vdots & \vdots& & \vdots \\ 0 & 0 &
    \cdots & 0\end{pmatrix}
\quad \mbox{ and } \quad
Y \assign \begin{pmatrix} A_1 & 0 & \cdots & 0 \\
    A_2 & 0 & \cdots & 0 \\  \vdots & \vdots& & \vdots \\ A_m & 0 & \cdots & 0 \end{pmatrix}
\]
yields
\begin{align*}
XY &= \begin{pmatrix} A_1 +A_2 + \ldots + A_m & 0 & \cdots & 0
    \\ 0 & 0 & \cdots & 0 \\ \vdots & \vdots& & \vdots \\ 0 & 0 &
    \cdots & 0\end{pmatrix} \quad \mbox{ and}\quad
YX = \begin{pmatrix} A_1 & A_1 A_2 & \cdots &
    A_1 A_m \\  A_2 A_1 & A_2 & \cdots & A_2 A_m \\  \vdots & \vdots&
    \ddots& \vdots \\ A_m A_1 & A_m A_2 & \cdots & A_m \end{pmatrix}
\end{align*}
The matrix $YX$ can be additively decomposed into $m$ matrices
according to the following pattern
\[
YX= \begin{pmatrix} * &  &  &
     & \\   & * & &  &   \\   &  & \ddots & &  \\  &  & & * &  \\
    &  &  & & * \end{pmatrix}
+ \begin{pmatrix} 0 & * &  &
     & \\   & 0 & &  &   \\   &  & \ddots & \ddots
    &  \\  &  & & 0 & * \\
    * &  &  & & 0 \end{pmatrix}
+\;\ldots\;
+ \begin{pmatrix} 0 & &  &
     & * \\  * & 0 & &  &   \\   &\ddots  & \! \ddots &  &  \\  
    &  & & 0 &  \\
    &  &  & * & 0 \end{pmatrix}
\]
where the $*$ stand for entries of $YX$ and for $i=1,\ldots,m$ the $i$th
star-pattern after the diagonal pattern is obtained by $i$ cyclic shifts of the columns
of the diagonal pattern. 

%As in the proof of Proposition~\ref{prop:norm2}, 

$XY$ and $YX$ are Hermitian and thus we can apply Lemma~\ref{lem:ineq}.
Then, by applying the triangle inequality, the unitary invariance of the operator norm and the facts that for all $j \neq k : \|A_j\|=1$, $\|A_j A_k\|=\|A_k A_j\|$, we obtain the desired statement.
\end{proof}

\chapter{Proof of \reflem{bqsmchain}}
\label{app:bqsmchainproof}
To prove \reflem{bqsmchain}, we need to introduce some more tools. %technical tools. %will use the following plemma.

The following proposition guarantees that the ``averaging property'' of the guessing probability (which holds by definition in the classical case) still holds when additionally conditioning on a quantum system.
\begin{prop}
\label{prop:avgprop}
For any state $\rho_{XYE}\in \sdm(\H_X\kron \H_Y \kron \H_E)$ that is classical on $X$ and $Y$ it holds that
\[
\gs(X|YE) = \sum_y P_Y(y)\, \gs(X|E,Y=y).
\]
\end{prop}
\begin{proof}
First, note that for any matrix $M_x$ acting on $\H_Y \kron \H_E$, we can always write
$M_x = \sum_{y,y'} \ketbra{y}{y'} \kron M_x^{y,y'}$, where $M^{y,y'}_x$ acts on $\H_E$ for every $x,y,y'$. Now, we write 
\begin{align*}
\gs(X|YE) & = \max_{\set{M_x}} \sum_x P_X(x) \trace(M_x \rho^x_{YE})\\
& = \max_{\set{M_x}} \sum_x P_X(x) \trace(M_x \sum_y P_{Y|X}(y|x) \, \proj{y}\kron \rho^{x,y}_{E})\\
%& = \max_{\set{M_x}} \sum_{x,y} P_{XY}(x,y) \trace(M_x (\proj{y}\kron \rho^{x,y}_{E}))\\
& = \max_{\set{M_x}} \sum_{x,y} P_{XY}(x,y) \trace( ( \sum_{v,w} \ketbra{v}{w} \kron M_x^{v,w} )(\proj{y}\kron \rho^{x,y}_{E}))\\
%& = \max_{\set{M_x}} \sum_{x,y} P_{XY}(x,y) \trace(  \sum_{v,w} \ket{v}\braket{w}{y}\bra{y} \kron M_x^{v,w} \rho^{x,y}_{E})\\
& = \max_{\set{M_x}} \sum_{x,y} P_{XY}(x,y)  \sum_{v}\braket{v}{y}\trace(   M_x^{v,y} \rho^{x,y}_{E})\\
& = \max_{\set{M_x}} \sum_{x,y} P_{XY}(x,y)  \trace(   M_x^{y,y} \rho^{x,y}_{E})\\
& = \sum_y P_Y(y) \max_{\set{M^{y,y}_x}} \sum_{x} P_{X|Y}(x|y) 
\trace(   M_x^{y,y} \rho^{x,y}_{E})\\
%
%\trace( (\proj{y}\kron M'_x)( \proj{y}\kron \rho^{x,y}_{E}))\\
%& = \sum_y P_Y(y) \max_{\set{M'_x}} \sum_{x} P_{X|Y}(x|y) \trace( M'_x \rho^{x,y}_{E})\\
&= \sum_y P_Y(y) \,\gs(X|E,Y=y).
\end{align*}
%where $M_x = \sum_{y,y'} \ketbra{y}{y'} \kron M_x^{y,y'}$ and acts on $\H_Y \kron \H_E$ for every $x$ and $M^{y,y'}_x$ acts on $\H_E$ for every $x,y,y'$.
\end{proof}

The following proposition is known as the chain rule for min-entropy.
\begin{prop}[\cite{Renner05}]
\label{prop:chain}
The following holds for all $\rho_{ABC} \in \sdm(\H_A \kron \H_B \kron \H_C)$,
\[
\hmin(A|BC) \geq \hmin(AB|C) -\hmax(B).
\]
\end{prop}
Finally, we need the following lemma.
\begin{lemma}
\label{lem:moreminent}
For any state $\rho_{XYE}\in \sdm(\H_X\kron \H_Y \kron \H_E)$ that is classical on $X$ and $Y$ it holds that
\begin{equation}
\hmin(XE|Y=y)\geq \hmin(X|Y=y) 
\label{eq:eigineq}
\end{equation}
for every $y \in \mcal{Y}$. 
\end{lemma}
\begin{proof}
Note that it suffices to show that
$\lambda_{\max}(\rho^y_{XE}) \leq \lambda_{\max}(\rho^y_{X})$ holds for every $y\in \mcal{Y}$. Because $\rho_{XE}^y$ is classical on $X$, there exists a unitary $U$ acting on $\H_X$ such that $\tilde{\rho}_{XE}^y := (U\kron\id_E) \rho_{XE}^y (U^\dagger \kron \id_E)$ is classical with respect to the computational basis $\set{\ket{x}}_{x \in \mcal{X}}$ on $\H_X$ with $\mcal{X}:=[d]$. In particular, this means that $\tilde{\rho}_{XE}^y$ has block-diagonal structure:
\[
\tilde{\rho}^y_{XE} = \sum_{x \in [d]} P_{X|Y}(x|y) \proj{x} \kron \rho_{E}^{x,y}=\begin{bmatrix} P_{X|Y}(1|y)\,\rho_E^{1,y} &&\boldsymbol{0}\\
&\ddots & \\
\boldsymbol{0} && P_{X|Y}(d|y)\,\rho_E^{d,y}\end{bmatrix}.
\]
Note that because $U$ is unitary, $\tilde{\rho}^y_{XE}$ has the same eigenvalues as $\rho^y_{XE}$, where these eigenvalues are given by the union of the eigenvalues of the blocks on the diagonal of $\tilde{\rho}^y_{XE}$. From this we see that the largest eigenvalue of $\tilde{\rho}^y_{XE}$ (and thus of $\rho^y_{XE}$) cannot be larger than the largest eigenvalue of $\tilde{\rho}^y_{X}:=\trace_E (\tilde{\rho}^y_{XE})$ (and thus of $\rho^y_{X}$). 
\end{proof}
\begin{proof}[Proof of \reflem{bqsmchain}]
By \refeq{guessform} it is equivalent to show that
\[
\gs(X|YE) \leq \gs(X|Y) \, 2^{\hmax(E)}. 
\]
Using \refprop{avgprop}, we write
\begin{align*}
\gs & (X|EY)  = \sum_y P_Y(y)\, \gs(X|E,Y=y) = \sum_y P_Y(y) \,2^{-\hmin(X|E,Y=y) } \\
&\leq \sum_y P_Y(y) \, 2^{- (\hmin(XE|Y=y)-\hmax(E))} \\
&\leq 2^{\hmax(E)} \, \sum_y P_Y(y) 2^{-\hmin(X|Y=y)} 
=2^{\hmax(E)} \, \gs (X|Y), 
\end{align*}
where the first inequality is \refprop{chain}, 
and the second inequality follows by \reflem{moreminent}. 
Hence, the claim follows.
\end{proof}
\chapter{Proof of \reflem{probsymmetry}}
\label{app:prfprobsym}
\begin{proof}
Let  $\alpha,\beta \in \cnum$ be such that $\theta_i:=\set{\alpha \ket{0}+\beta \ket{1}, \beta \ket{0}-\alpha \ket{1} }$. (We can always find such $\alpha$ and $\beta$.) Writing out the measurement explicitly gives
%\[
\begin{align*}
%\nonumber
P_{Y_i|X_i B_i \Theta_i}(0  | x_i , b_i, \theta_i)   &= | (\alpha \bra{0}+\beta \bra{1})H^{b_i} \ket{x_i}|^2
%\end{align}
\qquad\text{and}\\
%\begin{align}
%\nonumber 
P_{Y_i|X_i B_i \Theta_i}(1  | x_i , b_i, \theta_i)  &= | (\beta \bra{0}-\alpha \bra{1})H^{b_i} \ket{x_i}|^2. %(\alpha \bra{0}+\beta 
\end{align*}
Hence, it suffices to prove that
\begin{equation}
| (\alpha \bra{0}+\beta \bra{1})H^{b_i}\ket{x_i} |^2
= | (\beta \bra{0}-\alpha \bra{1})H^{b_i}\ket{x_i \xor 1 } |^2 
\label{eq:bsc}
\end{equation}
for every $ x_i,b_i \in \set{0,1}$.

We first show \refeq{bsc} for $b_i=0$. Let $\sigma_1$ be the first Pauli matrix defined by $\sigma_1 \ket{a} = \ket{a\oplus 1}$ for every $a \in \set{0,1}$. It follows immediately from the definition that $\sigma_1$ is a unitary matrix and it is easy to see that $\sigma_1$ is Hermitian. Then,
\begin{align*}
|(\alpha \bra{0}+\beta\bra{1})\ket{x_i}|^2 &= |(\alpha \bra{0}+\beta\bra{1}) \sigma_1 \sigma_1 \ket{x_i }|^2 = |(\alpha \bra{1} + \beta \bra{0}) \ket{x_i\xor 1}|^2 \\
&= |(\beta \bra{0}- \alpha \bra{1} ) \ket{x_i\xor 1}|^2 %= |\braket{w}{x}|^2
\end{align*}
The last equation follows because the expression equals either $|\alpha|^2$ or $|\beta|^2$ (depending on $x_i\in \set{0,1}$), hence we may freely change the sign of $\alpha$.  
For $b_i=1$, we have
\[
|(\alpha \bra{0}+\beta\bra{1})H\ket{x_i}|^2 = |(\alpha \bra{0}+\beta\bra{1})(\ket{0} + (-1)^{x_i}\ket{1})|^2 = |\alpha +(-1)^{x_i}\beta |^2
\]
and
\[
|(\beta \bra{0}-\alpha \bra{1})H\ket{x_i\xor 1}|^2 = |(\beta \bra{0}-\alpha \bra{1})(\ket{0} - (-1)^{x_i}\ket{1})|^2 = |\beta + (-1)^{x_i} \alpha|^2.
\]
We see that those expressions are equal for every $x_i \in \set{0,1}$.
\end{proof}

\end{appendix}
\end{document}